\documentclass[
aps,
floatfix, 
notitlepage, 
nofootinbib,
10pt, 
]{revtex4-2}

\usepackage[utf8]{inputenc}

\usepackage{amsmath}
\usepackage{amsthm}
\usepackage{amssymb}
\usepackage{braket}
\usepackage{color}
\usepackage{graphicx}
\usepackage{dsfont}
\usepackage{subcaption}
\usepackage{ragged2e}
\usepackage{tikz}
\usepackage{multirow}
\usepackage{hyperref}
\hypersetup{colorlinks,allcolors=black,linktoc=all}

\newtheorem{lemma}{Lemma}

\newtheorem{proposition}{Proposition}
\newtheorem{definition}{Definition}
\newtheorem{corollary}{Corollary}
\newtheorem{assumption}{Assumption}

\DeclareMathOperator{\supp}{supp}

\begin{document}
\title{On the stability to noise of fermion-to-qubit mappings}
\author{Guillermo González-García$^{1}$}
\author{Filippo Maria Gambetta$^{1}$}
\author{Raul A. Santos$^{1}$}
\email{raul@phasecraft.io}
\address{$^1$Phasecraft Ltd.}
\date{\today}

\begin{abstract}
    Quantum simulations before fault tolerance suffer from the intrinsic noise present in quantum computers. In this regime, extracting meaningful results greatly benefits from stability against that noise. This stability, defined as an error in observables that is independent of the system's size, is expected in local systems under local noise. In fermionic systems, the encoding of the fermionic degrees of freedom into qubits can introduce non-locality, making stability more delicate.
    Here, we investigate the stability to noise of fermion-to-qubit mappings. We consider noisy quantum circuits in $D$ dimensions modeled by alternating layers of local unitaries and general, single-qubit Pauli noise. We show that, when using local fermionic encodings, expectation values of quadratic fermionic observables are stable to noise in states with spatially decaying correlations: a power-law decay with exponent $\mu>D$ is sufficient for stability. By contrast, we show that this stability cannot be achieved by non-local encodings such as Jordan-Wigner in $2D$, or quasi-local ones such as the Bravyi-Kitaev transform. Our findings formalize the intuition that decaying correlations of the physical systems under study provide protection against noise for local fermionic encodings, and  help inform design principles in near-term quantum simulations.

\end{abstract}
\maketitle

\section{Introduction}

The recent advances in the performance of quantum hardware \cite{quantinuum2025helios,Google2025_QEC,quera2025QEC,xuquera2024constant,google2019quantumsupremacy,ibm2024qec,pan_2020_boson_sampling_exp,quantum-advantage-pan}, have accelerated the search for suitable quantum applications. However, current quantum devices still suffer from significant noise and decoherence, which severely limit their capabilities \cite{Preskill2018nisq}. In this context, significant research effort has been devoted to identifying potential applications for these noisy quantum devices \cite{eisert2025fasq}. Among them, the simulation of quantum physical systems stands out as particularly promising \cite{Daley2022_quantumsimulation,Cao2019_Quantum_simulations,Bauer2023_Quantum_simulation,Altman2021_Quantum_Simulation,Ebadi2021simulation}. Not only does it represent one of the most intuitive ways to use a quantum computer \cite{Feynman1982}, but in the noiseless setting, quantum simulation remains one of the few well-established exponential quantum speedups \cite{Montanaro2016_algorithms,lloyd1996_universal_quantum_simulators,Dalzell2025_algorithms_book}. In addition, there is a growing body of work that suggests that quantum simulation tasks may be particularly resilient to noise. Specifically, the concept of stability to noise has attracted significant attention, referring to  situations in which the noise-induced error in simulating local observables remains independent of the system size \cite{trivedi2024analog_stability,dreyer2025_dilution_of_errors,eisert2025stability,trivedi2025stability_mappings,Zoller2019_trotter,Cubitt2015stability,heyl2019_stability,Hastings2005stability,Bachmann2012_stability,schiffer2024proliferation}. This contrasts with the behavior of generic (e.g., random) circuits, for which errors tend to accumulate, which leads to an increasing impact of errors with system size \cite{Gonzalez2022_error_propagation,Quek2024_mitigation}. Consequently, quantum simulation tasks that exhibit stability to noise seem particularly well-suited for near-term devices: even with increasing system sizes, a constant hardware error-rate suffices to achieve a fixed target precision. Previous works have further suggested that such stability can lead  to exponential quantum speed-ups in the task of calculating local observables to a desired precision \cite{trivedi2025_accuracy_guarantees,trivedi2024analog_stability}.
Furthermore, this stability behavior has been identified in several settings, including both analog \cite{trivedi2024analog_stability,trivedi2025_accuracy_guarantees} and digital quantum simulations \cite{eisert2025stability,trivedi2025stability_mappings,dreyer2025_dilution_of_errors,Zoller2019_trotter}, as well as adiabatic algorithms \cite{schiffer2024proliferation}.

Recently, several large scale quantum simulation experiments have focused on fermionic systems \cite{phasecraft2025google,2025phasecraft_quantinuum,2025_phasecraft_fermionic_excitations,quantinuum2025superconducting,nigmatullin2025compact_vs_JW,Xu2025fermionic_simulation,brown2019_fermionic_simulation,sanchez2020_fermionic_simulation}. Many quantum systems of physical interest are naturally described by fermions, making the prospect of extracting their properties using quantum devices particularly compelling. However, most current hardware platforms use qubits as a basis, thus necessitating the use of fermion-to-qubit mappings in order to implement fermionic simulations \cite{JordanWigner1928,Derby2021_compact_encoding,derby2021compact_2,2005_Verstraete_cirac_encoding,Bravyi2002_encoding,troyer2017_encoding,kanav2019_encoding,Jiang2020encoding}. The earliest of such encodings is the Jordan-Wigner encoding, which represents a system of $N$ fermionic modes using $N$ qubits \cite{JordanWigner1928}. While conceptually simple, Jordan-Wigner often maps local fermionic operators to non-local, high-weight qubit operators. Subsequently, families of local encodings have been developed, which allow for mapping local fermionic operators to local qubit operators, with the cost of an increased qubit number overhead \cite{2005_Verstraete_cirac_encoding,Derby2021_compact_encoding,derby2021compact_2}.

In this context, it is of fundamental interest to understand the response to noise of such encodings. However, the results in this direction remain limited. Ref.~\cite{trivedi2024analog_stability} provided stability guarantees for certain classes of Gaussian fermionic systems, but focused on quantum simulators where fermions are implemented natively, without requiring fermion-to-qubit mappings. Other stability results have mostly focused on the simulation of spin systems \cite{trivedi2025stability_mappings,eisert2025stability}, for which fermionic mappings are likewise not required.  
This naturally raises the question of whether stability to noise is preserved under fermionic encodings. Moreover, given the wide range of available encodings now, each with different properties, it is natural to ask whether some encodings result in a more favorable response to noise than others. The latter question was partially addressed in Ref.~\cite{nigmatullin2025compact_vs_JW}, which experimentally compared the behavior of the compact \cite{Derby2021_compact_encoding,derby2021compact_2} and Jordan-Wigner encoding \cite{JordanWigner1928}, suggesting that the latter is significantly more fragile to noise. Additionally, Ref.~\cite{dyrenkova2025encodings} employed classical simulations to numerically compare the performance of several encodings in noisy settings. To the best of our knowledge, Ref.~\cite{dreyer2025_dilution_of_errors} contains the only analytical stability proof that includes fermion-to-qubit mappings, showing that the time evolution and posterior measurement of a $1-$dimensional non-interacting Fermi-Hubbard chain in the presence of single-qubit depolarizing noise is stable.

\subsection{Summary of results}
In this work, we rigorously study the stability to incoherent noise of fermion-to-qubit mappings in several settings. First, we consider the stability of the task of measuring quadratic observables in arbitrary encoded fermionic states in the presence of Pauli noise, where stability is defined as the regime in which the error in the observable is independent of system size. This serves as a simplified setting in which one can probe the response to noise of the fermionic encodings. In this setting, we derive a precise criterion for the stability of the measurement task. We rigorously show that, for a family of local fermionic encodings, sufficiently fast decaying correlations in the state suffices for stability. Precisely, for a $D-$dimensional system, a power-law correlation decay with respect to the distance with exponent $\mu>D$ (i.e., $|\langle  c_{\mathbf{r}}^{\dagger} c_{\mathbf{r}'}\rangle| \leq 1/d(\mathbf{r},\mathbf{r}')^{\mu}$) is sufficient for general Pauli noise, with the total error scaling as $f(p)$ for a noise rate $p$, with $f(p)=\mathcal{O}(p)$ when $\mu>D+1$ and $f(p)=\mathcal{O}(p^{\mu-D})$ when $D<\mu<D+1$. Thus, this analysis shows decaying correlations serve as a protection against noise for local fermionic encodings, and as a consequence stability can be achieved in $D-$dimensional systems for a wide range of physically relevant states, including thermal states or ground states of gapped Hamiltonians \cite{Hastings2006_correlations}. This threshold indicates a sharp transition between stability to noise and fragility to noise for local encodings. Additionally, we show these results do not generalize to encodings which are not local. Precisely, these stability results are unobtainable for a Jordan-Wigner encoding in $2D$ with snake ordering, which we show to be asymptotically fragile to noise due to the longer Pauli strings present. The same holds true for quasi-local encodings, such as the Bravyi-Kitaev transform \cite{Bravyi2002_encoding}: even though the noise scales in a milder way than $2D$ Jordan-Wigner, in the absence of strict locality they cannot satisfy the stability criterion. Furthermore, we also apply our analysis to ground states of non-interacting Fermi gases in one and two spatial dimensions. These are strongly correlated states that do not fulfill the correlation decay required above, and present a Fermi surface. For these states, we find that the task of measuring the momentum occupation operator $n_{\mathbf{k}}$ is considerably fragile to noise when $\mathbf{k}$ is close to the Fermi momentum, while it remains stable to noise elsewhere. 


We emphasize that the results above concern the task of measuring expectation values, which in isolation does not model all available operations in a quantum device. The purpose of this analysis is to use a concrete framework for the study of stability of fermionic encodings against noise. In particular it is natural to expect that an encoding that is fragile in this idealised setting will perform poorly in a more realistic scenario, while a stable encoding has some chance of being useful. 
The goal is that the resulting insights inform design principles for more realistic implementations. Indeed, building on these results, we derive stability bounds for the more general problem of measuring quadratic observables in $D-$dimensional systems after a noisy quantum circuit, modeled as alternating layers of noiseless, local unitaries and Pauli noise channels. This result requires that correlations in the initial state decay sufficiently rapidly, specifically with exponent $\mu > D$, arising from the threshold condition derived for the measurement task. This implies that stability is achievable for tasks such as noisy state preparation, or the constant time dynamics of a geometrically local fermionic Hamiltonian. In particular, the latter can be understood as a generalization of previous stability results for fermionic dynamics in analog quantum simulators \cite{trivedi2024analog_stability} to the digital setting, where fermionic degrees of freedom are encoded into qubits.

In particular, we consider noisy quantum circuits of constant depth $d=\mathcal{O}(1)$, and derive a tighter stability bound for non-interacting (i.e. Gaussian) circuits, and a more conservative bound for interacting circuits. Precisely, for a noise rate $p$, we bound the error as $\mathcal{O}(f(p) d^{D+2})$ for the free fermion case, and $\mathcal{O}(f(p) d^{2D+1})$ for the interacting case. These results can be compared with the stability to noise in analog quantum simulators that implement local fermionic Hamiltonians. In the analog case, a scaling of $\mathcal{O}\left(pt^{D+1}\right)$ after an evolution time $t$ can be shown in the interacting case, and an improved scaling of $\mathcal{O}(pt)$ in the free fermion case was demonstrated in Ref. \cite{trivedi2024analog_stability}. Intuitively, the overhead due to the use of fermionic encodings results in a steeper scaling in the digital setting as compared to the analog one. Nonetheless, our results rigorously show that stability can still be achieved in the digital setting for constant depths $d=\mathcal{O}(1)$, thanks to the protective effect of correlation decay.


The paper is organized as follows: in section \ref{section:stability_encodings} we analyze the stability to noisy measurement of the encodings. This includes a precise stability criterion for local encodings, and an explicit proof for the fragility non-local and quasi-local encodings. Furthermore, we include an analysis of the stability to noise of $1D$ and $2D$ states with a Fermi surface, respectively. 
Finally, in section \ref{section:noisy_evolution} we analyze noisy quantum circuits, including a proof for the free fermion case in subsection \ref{subsection:noisy_evolution_free}, and  a looser one for the interacting case in subsection \ref{subsection:noisy_evolution_interacting}.

\subsection{Setting and notation}

We consider fermionic systems of $N$ modes arranged on a $D-$dimensional lattice $\Lambda=\{0,1,\dots,L-1\}^D$ with periodic boundary conditions. We will use $d(\mathbf{r},\mathbf{r'}):=\sum_{i=1}^D \min(|r_i-r_i'|,L-|r_i-r_i'|)$ to refer to the distance between $\mathbf{r}$ and $\mathbf{r'}$.

To describe the fermionic modes, we will use the usual  annihilation and creation operators satisfying the canonical anticommutation relations $ \{c_{\mathbf{r}}^{\dagger},c_{\mathbf{r'}}\}=\delta_{\mathbf{r},\mathbf{r'}}, \forall \mathbf{r},\mathbf{r'}$. Additionally, it is often convenient to work with the $2N$ Majorana operators defined as 
\begin{align}\label{eq:majoranas}
\gamma_{\mathbf{r}}^1:=\left(c_{\mathbf{r}}^{\dagger}+c_{\mathbf{r}}\right) \quad ; \quad \gamma_{\mathbf{r}}^2:=i\left(c_{\mathbf{r}}^{\dagger}-c_{\mathbf{r}}\right),
\end{align}
which obey the anticommutation relations $\left \{\gamma_{\mathbf{r}}^\alpha,\gamma_{\mathbf{r'}}^\beta\right\}=2\delta_{\mathbf{r},\mathbf{r}'}\delta_{\alpha,\beta}$. Furthermore, we will denote the correlation matrix of a state $\rho$ as 
\begin{align}\label{eq:correlation_matrix}
\Gamma_{\mathbf{r},\mathbf{r'}}^{\alpha,\alpha'}(\rho):=\frac{i}{2}\mathrm{tr} \left(\left[\gamma_{\mathbf{r}}^\alpha,\gamma_{\mathbf{r}'}^{\alpha'}\right] \rho\right).
\end{align}

We will also use the creation and annihilation operators in momentum space, denoted by
\begin{align}\label{eq:momentum_modes}
a_\mathbf{k}:=\frac{1}{\sqrt{N}}\sum_{\mathbf{r} \in \Lambda} e^{-i\mathbf{k}\cdot\mathbf{r} }c_{\mathbf{r}},
\end{align}
where the allowed momenta are chosen as $\mathbf{k_m}=(2\pi/L) \mathbf{m}$, with $\mathbf{m} \in \{-L/2,\dots,L/2-1\}^D$, for odd occupation number (i.e., periodic boundary conditions), and  $\mathbf{m} \in \{-L/2+1/2,\dots,L/2-1/2\}^D$ for even occupation number (i.e., anti-periodic boundary conditions).

In this paper, we will map fermionic states to qubits. We will define the noise processes at the level of qubits as follows. Let us denote a state with $N_{\mathrm{qubits}}$ qubits as $\rho$, and an observable $O$. We will consider a single-qubit Pauli noise channel of error rate $p$,

\begin{align}\label{eq:Pauli_noise}
T_p(\rho):=\left(1-\frac{3}{4}p\right)\rho+\frac{3}{4}p(\alpha_x X\rho X +\alpha_y Y\rho Y+\alpha_z Z\rho Z),
\end{align}

\noindent where $\alpha_x+\alpha_y+\alpha_z=1$. Note that the specific case $\alpha_x=\alpha_y=\alpha_z=1/3$ yields depolarizing noise of rate $p$ due to the additional $3/4$ factor. We denote by $\mathcal{N}_p(\cdot)$ the quantum channel that applies single-qubit Pauli noise with error rate $p$ to all qubits. That is, $\mathcal{N}_p:=T_p^{\otimes N_{\mathrm{qubits}}}$. 


Analogously to Ref.~\cite{trivedi2024analog_stability}, we will say that the task of simulating an observable $O$ is robust to noise if the error between the noiseless and noisy observable is independent of the system size. Precisely, this is

\begin{definition}\label{definition:stability}
Consider an initial state $\rho_0$, a noiseless quantum channel $\mathcal{E}$ and a noisy quantum channel $\mathcal{E}'$ with an associated error rate $p$. The expectation value is said to be stable against errors if

\begin{equation}
\mathrm{Error}=|\mathrm{tr}(\mathcal{E}(\rho_0)O)-\mathrm{tr}(\mathcal{E}'(\rho_0)O)| \leq f(p),
\end{equation}
\noindent for some continuous function $f(p)$ of the error rate, independent of $N$, such that $\lim_{p \to 0}f(p)=0$.
\end{definition}

Note that in the definition above we allow for some liberty in the choice of the ideal quantum simulation ($\mathcal{E}$) and noisy quantum simulation ($\mathcal{E}'$). In section \ref{section:stability_encodings} we will consider $\mathcal{E}=\mathrm{id}$ to be the identity and $\mathcal{E}'(\cdot)=\mathcal{N}_p(\cdot)$, which will serve to probe the stability properties of the encoding when measuring expectation values. In section \ref{section:noisy_evolution} we will consider $\mathcal{E}$ to be a noiseless quantum circuit, while the noisy process $\mathcal{E}'(\cdot)$ will be modeled as the unitary layers in $\mathcal{E}(\cdot)$ interspersed with layers of the noise $\mathcal{N}_p(\cdot)$. Hence, this constitutes a model for analyzing the stability to noise of digital quantum simulations. In both cases, we will denote the ideal expectation value as $\langle O\rangle_{\mathrm{noiseless}}:=\mathrm{tr}(\mathcal{E}(\rho_0)O)$, and the noisy one as $\langle O\rangle_{\mathrm{noisy}}:=\mathrm{tr}(\mathcal{E}'(\rho_0)O)$.

Throughout the paper, for an operator $O$, we will use $\|O\|_p:=\left[\mathrm{tr}\left(|O|^p\right)\right]^{1/p}$ to denote its Schatten $p$-norm. Additionally, we will often use $\|O\|=\|O\|_{\infty}$ to denote its operator norm (which is also its Schatten $\infty$-norm). On the other hand, for vectors, we will often use the usual $\ell^1-$norm $\|\mathbf{x}\|_1:=\sum_{i} |x_i|$ (also commonly referred to as Manhattan norm). For a Hermitian matrix $A$ with elements $A_{i,j}$, we will often invoke the inequality

\begin{align}\label{eq:norm_inequality}
\|A\| \leq\max_{j} \sum_i |A_{i,j}|,
\end{align}
\noindent which follows from Gershgorin circle theorem.
Finally, throughout the text we will frequently use the big $\mathcal{O}$ notation from computer science \cite{cormen2022introduction}, which is detailed in appendix \ref{appendix:notation}.

\section{Stability of encodings}\label{section:stability_encodings}

\subsection{Local encodings}\label{subsection:local_encodings}


We will here consider local fermionic encodings, and provide a precise criterion for stability. Precisely, we will focus on encodings that map Majorana bilinear operators of the form $\gamma_\mathbf{r}^{\alpha}\gamma_{\mathbf{r}'}^{\alpha'}$ to Pauli strings with support growing as the distance between $\mathbf{r}$ and $\mathbf{r}'$. This is formalized in the following assumption.

\begin{assumption}\label{assumption:encodings}
Consider a local fermionic encoding such that the Majorana bilinear $\gamma_\mathbf{r}^{\alpha}\gamma_{\mathbf{r}'}^{\alpha'}$ is mapped to a Pauli string with support $\varphi(\mathbf{r},\mathbf{r}')=\varphi_0+d(\mathbf{r},\mathbf{r'})$, with $\varphi_0=\mathcal{O}(1)$ a constant.
\end{assumption}

\noindent The above assumption will be fulfilled when the support of the Pauli string between $\mathbf{r}$ and $\mathbf{r}'$ is chosen along the shortest path. This condition can be fulfilled when using local encodings such as the compact encoding \cite{Derby2021_compact_encoding,derby2021compact_2} or the Verstraete-Cirac encoding \cite{2005_Verstraete_cirac_encoding}. We note that, in the $2D$ case, a standard Jordan-Wigner with snake ordering will not fulfill this condition, since the Jordan-Wigner lines can reach $\Theta(N)$ length. This case will be handled in subsection \ref{subsection:JW}. Likewise, quasi-local encodings which map local operators to $\mathcal{O}(\log N)$ weight Pauli strings will be handled in subsection \ref{subsection:quasi_local}.

Since Pauli noise is diagonal in the Pauli basis, Assumption \ref{assumption:encodings} directly implies that the action of the Pauli noise defined on Eq.~\ref{eq:Pauli_noise} on $\gamma_{\mathbf{r}}^{\alpha}\gamma_{\mathbf{r'}}^{\alpha'}$ can be written as

\begin{align}\label{eq:dep_noise_1}
\mathcal{N}_p (\gamma_{\mathbf{r}}^{\alpha}\gamma_{\mathbf{r'}}^{\alpha'})=\lambda_{\mathbf{r},\mathbf{r'}}^{\alpha,\alpha'}\gamma_{\mathbf{r}}^{\alpha}\gamma_{\mathbf{r'}}^{\alpha'}.
\end{align}

\noindent Furthermore, since the bilinear $\gamma_{\mathbf{r}}^{\alpha} \gamma_{\mathbf{r'}}^{\alpha'}$ is mapped to a Pauli string with Pauli weight $\varphi(\mathbf{r},\mathbf{r}')$, we can bound the eigenvalues $\lambda_{\mathbf{r},\mathbf{r'}}^{\alpha,\alpha'}$ as 

\begin{align}\label{eq:bound_eigenvalues}
\left(1-\frac{3}{2}p\right)^{\varphi(\mathbf{r,\mathbf{r'}})} \leq \lambda_{\mathbf{r},\mathbf{r}'}^{\alpha, \alpha'} \leq 1.
\end{align}

\noindent  We will also assume that $p\leq 2/3$, which implies $\lambda_{\mathbf{r},\mathbf{r}'}^{\alpha, \alpha'} \geq 0$. 

We will consider general quadratic observables, which can be written as

\begin{align}\label{eq:expression_Ok_2D}
O=i\sum_{\mathbf{r},\mathbf{r}'} \sum_{\alpha, \beta}O_{\mathbf{r},\mathbf{r}'}^{\alpha, \beta} \gamma_{\mathbf{r}}^{\alpha} \gamma_{\mathbf{r}'}^{\beta},
\end{align}

\noindent and normalized as $\|O\| \leq 1$. Precisely, let us denote a $2N \times 2N$ matrix $\widetilde{O}$ with $\widetilde{O}_{(\mathbf{r},\alpha),(\mathbf{r}',\beta)}:=O_{\mathbf{r},\mathbf{r}'}^{\alpha, \beta}$. The matrix of coefficients $\widetilde{O}$ is chosen to be real and antisymmetric, which guarantees hermiticity for $O$. Furthermore, we will normalize it as $\|\widetilde{O}\|_1 \leq 1$, which directly implies $\|O\| \leq 1$ \footnote{This can be checked by block diagonalizing $\widetilde{O}$, and applying the triangle inequality to the eigenvalues of $\widetilde{O}$.}. Without loss of generality, we also assume that $O$ is traceless, imposing $O_{\mathbf{r},\mathbf{r}}^{\alpha,\alpha}=0$. 

From Eq.~\ref{eq:dep_noise_1}, the action of the Pauli noise channel on $O$ can be written as

\begin{align}\label{eq:Np_Ok}
\mathcal{N}_p(O)=i\sum_{\alpha,\beta}\sum_{{\mathbf{r}},
{\mathbf{r'}}}\left(\lambda_{\mathbf{r},\mathbf{r'}}^{\alpha,\beta}\right)O_{\mathbf{r},\mathbf{r'}}^{\alpha, \beta}\gamma_{\mathbf{r}}^{\alpha}\gamma_{\mathbf{r}'}^{\beta}
.
\end{align}

\noindent Hence, using the definition of the correlation matrix in Eq.~\ref{eq:correlation_matrix}, the error becomes 
\begin{align}\label{error_expression_DD}
|\langle O\rangle_{\mathrm{noiseless}}-\langle O\rangle_{\mathrm{noisy}}|= \left |\sum_{\alpha,\beta } \sum_{{\mathbf{r}},{\mathbf{r'}}} O_{\mathbf{r},\mathbf{r'}}^{\alpha, \beta} \left[1-\lambda_{\mathbf{r},\mathbf{r'}}^{\alpha,\beta}\right]\Gamma_{\mathbf{r},\mathbf{r}'}^{\alpha,\beta}
\right|.
\end{align}

\noindent Using this expression, we will show that the task of measuring quadratic observables is robust to noise when the correlations of the state $\rho_{0}$ decay sufficiently fast. This is stated formally in the following Proposition.

\begin{proposition}\label{proposition:stability_2D}
Consider an arbitrary fermionic state $\rho$ on a $D-$dimensional lattice. Assume that the correlations decay as
\begin{equation}
\left| \Gamma_{\mathbf{r},\mathbf{r}'}^{\alpha,\beta}(\rho) \right| \leq \frac{K}{d(\mathbf{r},\mathbf{r'})^{\mu}}, \quad\forall (\mathbf{r} \neq\mathbf{r'},\alpha,\beta)
\end{equation}

\noindent for some constant $K=\mathcal{O}(1)$, where $d(\mathbf{r},\mathbf{r'})$ represents the distance between both sites. Consider also a quadratic observable $O$, and the Pauli noise channel of noise rate $p$ in Eq.~\ref{eq:Pauli_noise}. Then, the task of measuring the expectation value of $O$ is stable to Pauli noise when $\mu>D$, with
\begin{align}
\mathrm{Error}=|\langle O\rangle_{\mathrm{noiseless}}-\langle O\rangle_{\mathrm{noisy}}| \leq f(p)= 3p\varphi_0+2KC_D\left(\zeta(\mu-D+1)-\left(1-\frac{3p}{2}\right)^{\varphi_0}\mathrm{Li}_{\mu-D+1}\left(1-\frac{3p}{2}\right)\right),
\end{align}

\noindent where $C_D=\mathcal{O}(1)$ is a dimension dependent constant, $\zeta(z)$ represents the Riemann zeta function, and $\mathrm{Li}_s(z)$ represents the polylogarithm of order $s$. Furthermore, in the small $p$ limit,


\begin{align}
f(p) = 
\begin{cases} 
\mathcal{O}\left( p^{\mu-D} \right) & \text{if } D < \mu < D+1, \\ 
\mathcal{O}\left( p \log \frac{1}{p} \right) & \text{if } \mu = D+1, \\ 
\mathcal{O}(p) & \text{if } \mu > D+1 ,
\end{cases}
\end{align}
\noindent ensuring $\lim_{p \to  0} f(p)=0$.

\end{proposition}

\begin{proof}

Using Eq.~\ref{error_expression_DD}, we write 
\begin{align}\label{eq:error_bound_Ok}
|\langle O\rangle_{\mathrm{noiseless}}-\langle O\rangle_{\mathrm{noisy}}| =   \left|\sum_{\alpha,\beta} \sum_{\mathbf{r}, \mathbf{r'}}O^{\alpha,\beta}_{\mathbf{r},\mathbf{r}'} \left[1-\lambda_{\mathbf{r},\mathbf{r'}}^{\alpha,\beta}\right]\Gamma^{\alpha,\beta}_{\mathbf{r},\mathbf{r}'}\right|=  \left|\sum_{\alpha,\beta} \sum_{\mathbf{r}, \mathbf{r'}}\widetilde{O}_{(\mathbf{r},\alpha),(\mathbf{r}',\beta)}F_{(\mathbf{r},\alpha),(\mathbf{r}',\beta)}\right| =  \left| \mathrm{tr}(\widetilde{O}F)\right|,
\end{align}

\noindent where $\widetilde{O}$ and $F$ are matrices of size $2N \times 2N$ with indices $(\mathbf{r},\alpha)$ and $(\mathbf{r}',\beta)$, defined as
\begin{align}
\widetilde{O}_{(\mathbf{r},\alpha),(\mathbf{r}',\beta)} := O^{\alpha,\beta}_{\mathbf{r},\mathbf{r'}} \quad ; \quad F_{(\mathbf{r},\alpha),(\mathbf{r}',\beta)}:=\left[1-\lambda_{\mathbf{r},\mathbf{r'}}^{\alpha,\beta}\right]\Gamma^{\alpha,\beta}_{\mathbf{r},\mathbf{r}'},
\end{align}

\noindent respectively. Hence, applying Hölder's inequality to Eq.~\ref{eq:error_bound_Ok}, we obtain
\begin{align}\label{eq:error_bound_MF}
|\langle O\rangle_{\mathrm{noiseless}}-\langle O\rangle_{\mathrm{noisy}}| = \left|\mathrm{tr} (\widetilde{O}F)\right|\leq   \|\widetilde{O}\|_1 \|F\|.
\end{align}

\noindent Note that $\|\widetilde{O}\|_1 \leq 1$ due to the normalization chosen in Eq.~\ref{eq:expression_Ok_2D}. Furthermore, since $F$ is hermitian, its spectral norm $\|F\|$ can be bounded using Eq.~\ref{eq:norm_inequality} as
\begin{align}\label{norm:F}
\|F\| \leq \max_{(\mathbf{r'},\beta)} \sum_{(\mathbf{r},\alpha)} \left| F_{(\mathbf{r},\alpha),(\mathbf{r'},\beta)}\right| .
\end{align}

\noindent Let use denote $r=1-3p/2$. Then, from Eq.~\ref{eq:bound_eigenvalues} we know that $\lambda_{\mathbf{r},\mathbf{r}'}^{\alpha,\beta} \geq r^{\varphi(\mathbf{r},\mathbf{r'})}=r^{d(\mathbf{r},\mathbf{r}')+\varphi_0}$, and we can write

\begin{align}\label{eq:sum_F}
 \max_{(\mathbf{r'},\beta)} \sum_{(\mathbf{r},\alpha)} \left| F_{(\mathbf{r},\alpha),(\mathbf{r'},\beta)}\right| = \max_{(\mathbf{r'},\beta)} \sum_{(\mathbf{r},\alpha)}\left[1-\lambda_{\mathbf{r},\mathbf{r'}}^{\alpha,\beta}\right]\left|\Gamma^{\alpha,\beta}_{\mathbf{r},\mathbf{r}'}\right|
 \leq \max_{(\mathbf{r'},\beta)} \sum_{(\mathbf{r},\alpha)}\left[1-r^{d(\mathbf{r},\mathbf{r'})+\varphi_0}\right]\left|\Gamma^{\alpha,\beta}_{\mathbf{r},\mathbf{r}'}\right|.
\end{align}

The sum in Eq.~\ref{eq:sum_F} is bounded rigorously in Appendix \ref{appendix:sum}. Precisely, Corollary \ref{corollary:sum_1} implies that 

\begin{align}
\max_{(\mathbf{r'},\beta)} \sum_{(\mathbf{r},\alpha)}\left[1-r^{d(\mathbf{r},\mathbf{r'})+\varphi_0}\right]\left|\Gamma^{\alpha,\beta}_{\mathbf{r},\mathbf{r}'}\right| \leq  2(1-r)\varphi_0+2KC_D\left(\zeta(\mu-D+1)-r^{\varphi_0}\mathrm{Li}_{\mu-D+1}\left(r\right)\right),
\end{align}

\noindent Using Eq.~\ref{eq:error_bound_MF} allows us to bound the error expression as
\begin{align}\label{eq:error_momentum_Ok}
|\langle O\rangle_{\mathrm{noiseless}}-\langle O\rangle_{\mathrm{noisy}}| \leq f(r)= 2(1-r)\varphi_0+2KC_D\left(\zeta(\mu-D+1)-r^{\varphi_0}\mathrm{Li}_{\mu-D+1}\left(r\right)\right).
\end{align}
proving the Proposition. Note that in the $r \to 1$ (i.e $p\rightarrow 0)$ limit $f(r)$ behaves as 
\begin{align}\label{eq:scaling_f(p)}
f(r) =  
\begin{cases} 
\mathcal{O}\left( (1-r)^{\mu-D} \right) & \text{if } D < \mu < D+1, \\ 
\mathcal{O}\left( (1-r) \log \frac{1}{1-r} \right) & \text{if } \mu = D+1, \\ 
\mathcal{O}(1-r) & \text{if } \mu > D+1.
\end{cases}
\end{align}
\end{proof}

We therefore find that the measurement of expectation values of quadratic observables is stable against noise for states whose correlations decay sufficiently fast. This includes a wide range of physical states, such as thermal states, ground states of gapped Hamiltonians, and a broad class of gapless, interacting systems. The intuition underlying this result is as follows. Any quadratic observable $O$  can be expressed as a linear combination of Majorana bilinears, as in Eq.~\ref{eq:expression_Ok_2D}. In a local encoding, each Majorana bilinear $ \gamma_{\mathbf{r}}^{\alpha}\gamma_{\mathbf{r}'}^{\beta}$ is mapped to a Pauli string whose Pauli weight scales with the distance $ d(\mathbf{r},\mathbf{r}')$. When this distance is small, the corresponding Pauli string is supported on only a few qubits and is therefore relatively resilient to noise. Conversely, bilinears associated with large distances are mapped to highly nonlocal Pauli strings with large Pauli weight. However, the contribution of such terms is suppressed by the decay of correlations in the state, which ensures overall stability. Making this intuition precise leads to the Proposition stated above.

This result implies that, even though the noise is effectively non-local in the fermions, the presence of correlation decay protects the system against this non-locality. This also suggests that, for states which present correlation decay in real space, a fermionic encoding in real space might be a better choice than one in momentum space, since the latter might not enjoy this protection against errors.

\subsection{Jordan-Wigner encoding}\label{subsection:JW}

Here, we will show that the stability to noise in subsection \ref{subsection:local_encodings} does not apply to Jordan-Wigner encodings with snake ordering in 2D. Intuitively, the longer strings in Jordan-Wigner encodings result in fragility to noise, even with strong correlation decay. This result is consistent with previous results \cite{nigmatullin2025compact_vs_JW}, which experimentally find that the Jordan-Wigner encoding is significantly more fragile to noise. To show this, we will consider an $L \times L$ lattice with snake ordering, and a state $\rho$ with exponentially decaying correlations, such that
\begin{align}
	|\langle c_{\mathbf{r}}^{\dagger}c_{\mathbf{r'}}\rangle_{\rho}| \leq Ke^{-\frac{d(\mathbf{r},\mathbf{r'})}{\xi}}.
\end{align}

\noindent  Due to the snake ordering, the Jordan-Wigner string for a hopping term between vertical neighbors $\mathbf{r}=(x,y)$ and $\mathbf{r'}=(x,y+1)$ has length $L+1$.  Consider the observable $O=\gamma_{\mathbf{r}}^{\alpha}\gamma_{\mathbf{r}'}^{\beta}$, and let us further assume that $\langle O\rangle=1$. We will consider a depolarizing noise channel, which implies choosing $\alpha_x=\alpha_y=\alpha_z=1/3$ in Eq.~\ref{eq:Pauli_noise}. Then, the noisy expectation value is given by 
\begin{align}
	\left \langle O\right \rangle_{\mathrm{noisy}} =\left  \langle \mathcal{N}_p \left(\gamma_{\mathbf{r}}^{\alpha} \gamma_{\mathbf{r'}}^{\beta}\right)\right \rangle  = (1-p)^{L+1} \langle \gamma_{\mathbf{r}}^{\alpha} \gamma_{\mathbf{r}'}^{\beta}\rangle=(1-p)^{L+1},
\end{align}

\noindent while the ideal expectation value is trivially $\left \langle O\right \rangle_{\mathrm{noiseless}}=1$. Since $L = \sqrt{N}$, it is then clear that the measuring of $O$ cannot be stable to noise. Precisely, the error will scale as

\begin{align}
\mathrm{Error}=\left | \left \langle O\right \rangle_{\mathrm{noisy}}-\left \langle O\right \rangle_{\mathrm{noiseless}}\right|=1-(1-p)^{L+1} = 1-(1-p)^{\Theta(\sqrt{N})},
\end{align}

\noindent and hence the error will grow towards $1$ exponentially fast with $\sqrt{N}$.

Therefore, we conclude that an analogous result to that of subsection \ref{subsection:local_encodings} is not possible for Jordan-Wigner encodings in $2D$: even for states with exponentially decaying corelations, there are quadratic observables that are not robust to noise. This arises from the fact that, using the Jordan-Wigner encoding, there can be local fermionic observables which are mapped to qubit operators with long Jordan-Wigner strings. 

\subsection{Quasi-local encodings}\label{subsection:quasi_local}

While section \ref{subsection:JW} focused on the Jordan-Wigner encoding in $2D$, which is manifestly non-local, there exists an important class of encodings that are quasi-local, meaning that local fermionic operators are mapped to Pauli strings with Pauli weight that can scale logarithmically with the system size, $\mathcal{O}(\log N)$. A prominent example is the Bravyi-Kitaev transform \cite{Bravyi2002_encoding}, as well as generalized tree-based mappings such as ternary tree encodings \cite{Jiang2020encoding}. Here, we show that this logarithmic scaling breaks the stability condition outlined in Definition \ref{definition:stability}, leading to an error that grows polynomially with $N$.

Precisely, we will consider encodings for which the number operator $n_i=c_i^{\dagger}c_i$ is mapped as
\begin{align}
	n_i  \rightarrow \frac{1}{2}(I - Z_{\text{logic}, i}),
\end{align}
where $Z_{\text{logic}, i}$ is a string of physical Pauli $Z$ operators with Pauli weight $\mathcal{O}(\log N)$. Furthermore, we will assume that there is at least an index $i$ such that the weight grows as $\Theta(\log N)$. This is true for both the Bravyi-Kitaev transform and the ternary tree mappings \cite{Bravyi2002_encoding,Jiang2020encoding}.

Now, we pick an index $i$ with weight $|Z_{\mathrm{logic},i}| =\Theta (\log N)$, and we choose a state such that $\langle n_i \rangle_{\mathrm{noiseless}} = 1$. We can compute the action of the depolarizing noise as
\begin{align}
\langle n_i\rangle_{\mathrm{noisy}}=\frac{1}{2}+\frac{1}{2}(1-p)^{|Z_{\mathrm{logic},i}|},
\end{align}

\noindent and hence the error becomes
\begin{align}
\mathrm{Error}= \left| \langle n_i\rangle_{\mathrm{noisy}} - \langle n_i\rangle_{\mathrm{noiseless}}\right|=\frac{1}{2}-\frac{1}{2}(1-p)^{\Theta(\log n)}.
\end{align}

\noindent This shows that the error will grow towards $1/2$ polynomially fast with $N$, and thus the stability as defined in Definition \ref{definition:stability} cannot be achieved. The slower growth of the error as compared to the Jordan-Wigner transform of subsection \ref{subsection:JW} arises naturally from the fact that the encoding is quasi-local.

\subsection{1D Fermi surface}\label{subsection:1D_Fermi}

We have thus far shown that, for local encodings fulfilling Assumption \ref{assumption:encodings}, states with sufficiently fast decaying correlations are stable to Pauli noise. Here we will show that our technique can also bound the stability to noise for strongly correlated states where correlations do not decay as fast. For simplicity, we will start with a $1D$ case. We remark that in $1D$ a standard Jordan-Wigner encoding fulfills Assumption \ref{assumption:encodings}. Hence, we will assume that Assumption \ref{assumption:encodings} is fulfilled with a constant $\varphi_0=1$, which is consistent with a $1D$ Jordan-Wigner encoding.

In this section, we will focus on a depolarizing noise channel, which implies choosing $\alpha_x=\alpha_y=\alpha_z=1/3$ in the Pauli noise channel Eq.~\ref{eq:Pauli_noise}. In what follows, we will focus on states which are translation invariant, and with definite particle number, which implies $\langle a_k a_q \rangle_{\rho}=\langle a_k^{\dagger} a_q ^{\dagger}\rangle_{\rho}=0$, and $\langle a_k^{\dagger} a_q \rangle_{\rho}=\delta_{k,q} \langle n_k \rangle_{\rho}$, where $a_k$ is the annihilation operator in momentum space, and $n_k=a_k^{\dagger}a_k$ denotes the number operator in momentum space.  Hence, for these states, the expectation value of any quadratic operator $O$ can be written as  $\langle O  \rangle=\sum_{k} \epsilon(k) \langle n_k \rangle$, for some coefficients $\epsilon(k)$. As a consequence, in the rest of the section we will focus on the error for the expectation value of the number operator $n_k=a_k^{\dagger}a_k$, which can be written as

\begin{align}\label{eq_error_nk_DD}
|\langle n_{{k}}\rangle_{\mathrm{noiseless}}-\langle n_{{k}}\rangle_{\mathrm{noisy}}|=\frac{1}{N} \left|\sum_{x,y} e^{ik(x-y)}\left\langle \left(\mathrm{id} - \mathcal{N}_p \right) c_x^{\dagger}c_y \right\rangle \right|.
\end{align}
We will analyze the error as a function of the momentum occupation function $n(q):=\langle a_q^{\dagger}a_q \rangle_{\rho}$ of a given state $\rho$. We will show rigorously that, in this case, the regularity of $n(q)$ controls the stability to noise. Note that this is consistent with the findings in subsection \ref{subsection:local_encodings}, since $n(q)$ is the Fourier transform of the correlation function $\langle c_x^{\dagger}c_y \rangle_{\rho}$, and hence regularity in $n(q)$ implies correlation decay. Precisely, smoothness in $n(q)$ implies exponentially decaying correlations, which automatically yields stability following the results in subsection \ref{subsection:local_encodings}. 

In the rest of the section, we will rigorously show that Lipschitz continuity for $n(q)$ is a sufficient condition for stability in 1D.\footnote{Note that Lipschitz continuity is slightly different from the correlation decay of Proposition \ref{proposition:stability_2D}: Lipschitz continuity merely guarantees correlation decay strictly faster than inversely with the distance, $|\langle c_x^{\dagger}c_y\rangle|  \leq o(1/d(x,y))$ \cite{Katznelson_2004_harmonic}, which is in general not sufficient to guarantee a decay with $\mu>1$ as stated in Proposition \ref{proposition:stability_2D}.} Additionally, we will also analyze states in which $n(q)$ contains discontinuities, which implies correlations decaying inversely with the distance, and thus remain out of reach for Proposition  \ref{proposition:stability_2D}. Precisely, we will consider states whose momentum occupation can be written as 

\begin{align}
n(q)=n_L(q)+\sum_{j=1}^M \Delta_j H(q-q_j),
\end{align}

\noindent where $n_L(q)$ is an $L-$Lipschitz continuous function, $H(q)$ is a Heaviside step function, and $|\Delta_j| \leq 1$. This expression includes momentum occupation functions which contain $M$ jump discontinuities at the points $q_j$, and is Lipschitz continuous everywhere else. Later, we will consider the specific case of a non-interacting Fermi gas with a Fermi surface, which contains only two discontinuities.

In the general setting, our results show that measuring momentum occupation is stable when measuring away from the discontinuities.

\begin{proposition}\label{proposition:result_discontinuities}
Consider a translation invariant, $N-$mode fermionic state $\rho$, and denote $n(q)=\mathrm{tr}(\rho a_q^{\dagger}a_q)$. Assume Assume that $n(q)$ is a piecewise defined function with $M$ jump discontinuities at the set of points $\mathcal{Q}=\{q_1,...,q_M\}$, and is $L-$Lipschitz continuous elsewhere. Then, the observable $n_k=a_{k}^{\dagger}a_{k}$ is stable against depolarizing errors if the distance between $k$ and any discontinuity is larger than a constant. Precisely, the error can be bounded as 
\begin{equation}
\mathrm{Error} = |\left\langle n_k\right\rangle_{\mathrm{noisy}}-\left\langle n_k\right\rangle_{\mathrm{noiseless}}|\leq f(p)= \mathcal{O}\left(\frac{pM}{q_{\mathrm{dist}}}+L\sqrt{p}\right),
\end{equation}
where $q_{\mathrm{dist}}:=\min_{ q_j \in \mathcal{Q}}|k-q_j|$ is the distance between $k$ and the set of discontinuities $\mathcal{Q}$, and $f(p)=\mathcal{O}(Mp+L\sqrt{p})$ when $q_{\mathrm{dist}}=\Omega(1)$.

\end{proposition}

\begin{proof}
Let us write the momentum occupation as 
\begin{align}
n(q)=n_L(q)+\sum_{j=1}^N \Delta_j H(q-q_j),
\end{align}

\noindent where $n_L(q)$ is Lipschitz continuous, and $H(q-q_j)$ is a Heaviside step function that models the jump discontinuities, with $|\Delta_j| \leq 1$. Note that, following Assumption \ref{assumption:encodings}, it follows that the noise acts as 
\begin{align}
\mathcal{N}_p (c_x^{\dagger}c_y)=(1-p)^{\varphi(x,y)}c_x^{\dagger}c_y+p\frac{I}{2}\delta_{x,y}.
\end{align}

\noindent Therefore, the action of the noise channel on the operator $n_k$ is given by

\begin{align}\label{eq:Np_nk_error}
\mathcal{N}_p(n_k)=\frac{1}{N}\sum_{x,y}e^{ik(x-y)}\mathcal{N}_p(c_x^{\dagger}c_y)=\frac{1}{N}\sum_{x,y}e^{ik(x-y)}(1-p)^{\varphi(x,y)}c_{x}^{\dagger}c_y+p\frac{I}{2}.
\end{align}

\noindent Since $\rho$ is assumed to be translation invariant, this implies that $\langle a_{p}^\dagger a_{q}\rangle_\rho =\delta_{pq}\langle n_q\rangle_{\rho}$. Taking this into account, and performing a Fourier transform in the terms $c_x^{\dagger}c_y$ in Eq~\ref{eq:Np_nk_error}, we can write the error as

\begin{align}\label{eq:error_discontinuous}
|\langle n_k\rangle_{\mathrm{noiseless}}-\langle n_k\rangle_{\mathrm{noisy}}| &=|\mathrm{tr}(n_k \rho)-\mathrm{tr}(\mathcal{N}_p(n_k)\rho)|=\frac{1}{N^2}\left|\sum_{x, y}\sum_q\left[1-(1-p)^{\varphi(x,y)}\right] e^{i(k-q)(x-y)}n(q)\right|+\frac{p}{2}\nonumber\\
&\leq |E_L(k)|+\sum_j \left|E_j^{\mathrm{jump}}(k)\right|+\frac{p}{2},
\end{align}
\noindent where we have defined
\begin{align}
E_L(k):=\frac{1}{N^2} \sum_{x,y}\sum_q\left[1-(1-p)^{\varphi(x,y)}\right] e^{i(k-q)(x-y)}n_L(q),
\end{align}

\noindent and 
\begin{align}
E_j^{\mathrm{jump}}(k):=\frac{\Delta_j}{N^2}\sum_{x,y}\sum_q\left[1-(1-p)^{\varphi(x,y)}\right] e^{i(k-q)(x-y)}H(q-q_j).
\end{align}

\noindent Then, the terms $E_L(k)$ and $E_j^{\mathrm{jump}}(k)$ can be bounded independently. Precisely, in Appendix \ref{appendix:lipschitz} we show that $|E_L(k)| \leq \mathcal{O}\left(L\sqrt{p}\right)$. Furthermore, we will use that
 \begin{align}\label{eq:error_Ej(k)}
|E_j^{\mathrm{jump}}(k)| \leq \mathcal{O}\left(\frac{p}{|k-q_j|}\right) \; \mathrm{if} \; k \neq q_j,
\end{align}

 \noindent and $|E_j^{\mathrm{jump}}(k=q_j)|=\mathcal{O}(pN)$. This upper-bound is rigorously shown in appendix \ref{appendix:jump_proof}. Intuitively, it implies that the error due to a jump discontinuity at momentum $q_j$ is $\mathcal{O}(1)$ when measuring a momentum $k$ that is away from $q_j$ by a constant quantity, $|k-q_j| \geq \Theta(1)$, but measuring momenta close to the discontinuity $q_j$ can yield large errors that are proportional to $N$. Utilizing Eq.~\ref{eq:error_Ej(k)} we can directly bound 

 \begin{align}
\sum_{j}|E_j^{\mathrm{jump}}(k)| \leq \mathcal{O}\left(\frac{pM}{q_{\mathrm{{dist}}}}\right).
 \end{align}
Plugging everything into Eq.~\ref{eq:error_discontinuous} completes the proof.

 \end{proof}

Following this result, one can conclude that the stability results are closely linked to the presence of discontinuities in $n(q)$: measuring a momentum close to a jump discontinuity will in general result to fragility to noise, and the total error will in general be proportional to the number of discontinuities. 

Now, we will apply this result to the specific case of non-interacting Fermi gases. In $1D$, the momentum occupation for these states is given by 
\begin{align}
n(q)=\begin{cases} 
1 &\text{if} \quad |q| \leq k_F\\
0 &\text{if} \quad |q|>k_F
\end{cases},
\end{align}

\noindent where $k_F:= \pi\nu$ denotes the Fermi momentum. Here, $\nu:=N_{\rm occ}/N$ is the filling fraction of the system, with $N_{\rm occ}$ the occupation number. Following Proposition \ref{proposition:result_discontinuities}, we bound the error as

\begin{equation}
\mathrm{Error}=\left|  \langle n_{k}\rangle_{\mathrm{\mathrm{noiseless}}}-\langle n_{k}\rangle_{\mathrm{\mathrm{noisy}}} \right| \leq \mathcal{O}\left(\frac{p}{k-k_F}\right).
\end{equation}

\noindent Hence, the expectation value of $n_k$ is stable to noise unless the momentum $k$ is close to the Fermi momentum, in which case the error will grow with the system size $N$. We investigate this numerically in Fig. \ref{fig:Fermi_surface_1D}, which shows that momenta close to the Fermi surface exhibit significant fragility.

\begin{figure}[t]
    \captionsetup{justification=justified,singlelinecheck=false}
    \centering
    \begin{subfigure}[t]{0.48\textwidth}
        \centering
        \caption*{\textbf{(a)}}
        \includegraphics[width=0.8\linewidth]{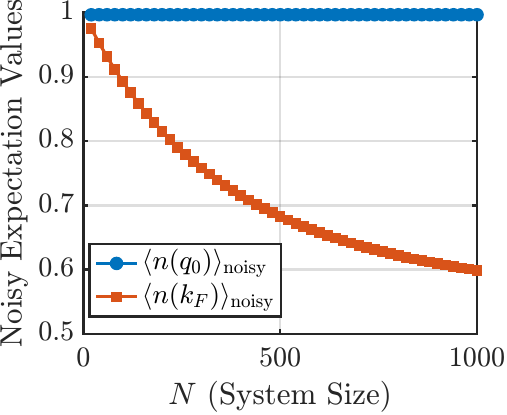}
    \end{subfigure}
    \hfill
    \begin{subfigure}[t]{0.48\textwidth}
        \centering
        \caption*{\textbf{(b)}}
        \includegraphics[width=0.8\linewidth]{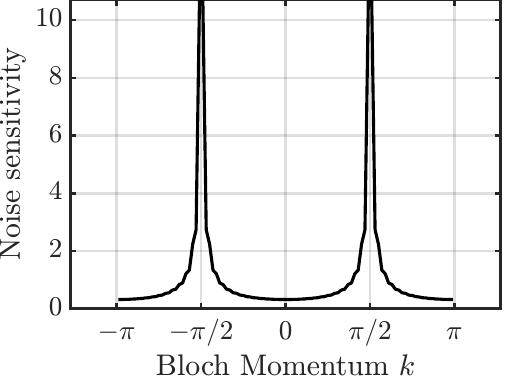}
    \end{subfigure}
    \caption{\justifying
        Response to noise of a $1D$ ground state with Fermi surface. (a) Expectation values of the momentum occupation with momenta $q_0=2\pi /N $ and Fermi momentum $k_F$, as a function of the system size with a fixed noise rate $p=10^{-2}$, and half-filling condition $(k_F=\pi/2)$. As predicted, the momentum expectation value at the Fermi surface rapidly degrades with increasing system size, signalling fragility to noise. On the other hand, the measurement of the occupation at momentum $q_0$ is stable to noise and does not deteriorate for large system sizes.
        (b) Representation of the sensitivity to noise of the noisy expectation value $\langle n_k\rangle_{\mathrm{noisy}}$ for a fixed system size $N=100$. The sensitivity to noise for a fixed momentum $k$ is estimated as the slope $| \langle n_k\rangle_{\mathrm{noisy}}-\langle n_k\rangle_{\mathrm{noiseless}}|/p$, where the noisy expectation value is computed for $p=10^{-2}$, and the Fermi momentum is $k_F=\pi/2$ (half-filling condition). As predicted by the theory, momenta close to the Fermi surface exhibit remarkable sensitivity to noise, while momenta further from the Fermi surface are stable to noise.
    }\label{fig:Fermi_surface_1D}
\end{figure}

\subsection{2D Fermi surface}\label{subsection:2D_Fermi}

In this section we will consider $2D$ states with a Fermi surface, and we will analyze their fragility to noise. Precisely, we will analyze a $2D$ tight-binding model with nearest-neighbor hopping, described by the Hamiltonian 
\begin{align}
H=-t\sum_{\langle i,j\rangle} (c_{i}^{\dagger}c_j+\mathrm{h.c.}).
\end{align}

\noindent For a square lattice, the corresponding dispersion relation is given by
\begin{align}
\epsilon(k_x,k_y)=-2t(\cos (k_x)+\cos (k_y)).
\end{align}

\noindent The ground state of this Hamiltonian exhibits a Fermi surface defined by the contour $\epsilon(k_x,k_y)=E_F$, where $E_F$ is the Fermi energy. The shape of the Fermi surface is well understood, and is a function of the occupation number $N_{\mathrm{occ}}$ \cite{simon2013oxford}. For occupation less than half-filling $(N_\mathrm{occ} <N/2)$, the Fermi surface is nearly circular. For half-filling occupation $(N_\mathrm{occ} =N/2)$, the Fermi surface becomes a diamond. Finally, for more than half-filling $(N_\mathrm{occ} >N/2)$, the Fermi surface becomes circular around the corners of the Brillouin zone. In Fig. \ref{fig:Fermi_surface_2D} we represent the fragility to noise of the measurements of momentum occupation for different momenta in the ground-state of the tight-binding model. Analogously to the $1D$ case, we find that the measurement of the momentum occupation exhibits significant fragility close to the Fermi surface, while remaining stable to noise elsewhere. Our numerical results show that the shape of the Fermi surface can be identified by measuring the fragility to noise. Precisely, Fig. \ref{fig:Fermi_surface_2D} recovers the correct Fermi surface shape for different occupation numbers. Finally, in Appendix \ref{appendix:2D_Fermi} we provide an analytical proof that confirms these numerical results: we study the thermodynamic limit of a system with a circular Fermi surface of radius $k_F$, and show that the measuring of expectation values is unstable exactly at the Fermi surface and stable elsewhere.

\begin{figure}[t]
 
    \includegraphics[width=0.95\linewidth]{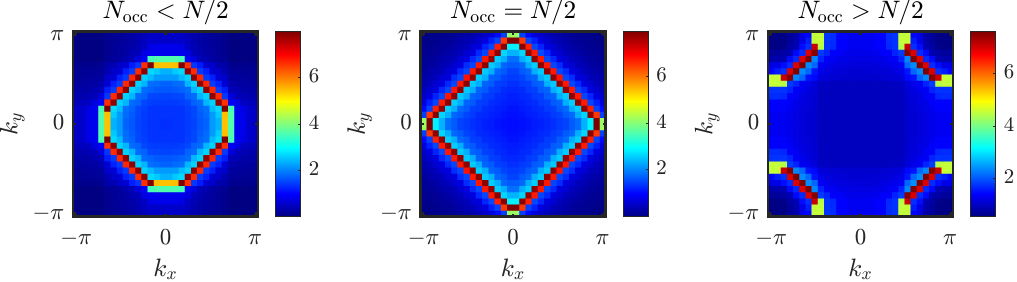}
 
    \caption{\justifying
        Response to noise of the ground state of a $2D$ tight-binding model with Fermi surface. We represent the sensitivity to noise of the noisy expectation value $\langle n(k_x,k_y)\rangle_{\mathrm{noisy}}$ for a lattice of size $N=30 \times 30$, and occupation number $N_{\mathrm{occ}}=300$ (left, less than half-filling), $N_{\mathrm{occ}}=450$ (center, half-filling), $N_{\mathrm{occ}}=700$ (right, more than half-filling). The sensitivity to noise for a fixed momentum mode $(k_x,k_y)$ is estimated as the slope $| \langle n(k_x,k_y)\rangle_{\mathrm{noisy}}-\langle n(k_x,k_y)\rangle_{\mathrm{noiseless}}|/p$, where the noisy expectation value is computed for $p=10^{-2}$. As predicted by the theory, the momenta close to the Fermi surface exhibit remarkable sensitivity to noise, while the momenta further from the Fermi surface are stable to noise.
    }\label{fig:Fermi_surface_2D}
\end{figure}
\section{Noisy quantum circuits}\label{section:noisy_evolution}

In this section we analyze the more general problem of stability for noisy quantum circuits.
Precisely, we consider constant depth $d=\mathcal{O}(1)$ circuits consisting of geometrically local gates, and analyze the task of measuring quadratic observables. We will consider fermionic states of $N$ modes in a $D-$dimensional lattice $\Lambda=\{0,1,\dots,L\}^D$, that are mapped to qubits via an encoding that fulfills Assumption \ref{assumption:encodings}.

We will consider an ideal circuit of constant depth $d=\mathcal{O}(1)$ given by
\begin{align}\label{eq:ideal_time_evolution}
\rho_{\mathrm{id}}(d)=\Phi_d^{\mathrm{id}} \Phi_{d-1}^{\mathrm{id}} \dots \Phi_1^{\mathrm{id}} \rho(0),
\end{align}

\noindent where each channel $\Phi_k^{\mathrm{id}}:=\mathcal{U}_k$ consists of a layer of unitary, local gates. We will also assume that the gates preserve fermionic parity, since the dynamics would otherwise be unphysical due to fermionic superselection rules.

We remark that the above circuit structure is very general, and can be applied to understand a broad class of computational tasks. In particular, this structure is applicable to the quantum simulation of fermionic dynamics, which has received considerable experimental interest in recent times \cite{phasecraft2025google,2025phasecraft_quantinuum,2025_phasecraft_fermionic_excitations,quantinuum2025superconducting,nigmatullin2025compact_vs_JW,Xu2025fermionic_simulation,brown2019_fermionic_simulation,sanchez2020_fermionic_simulation}. In the specific case of constant time ($t=\mathcal{O}(1)$) evolution of a Hamiltonian, the Hamiltonian dynamics can directly be mapped to circuits of the above form if the Hamiltonian can be implemented natively in the device. Naturally, this does not hold in general, since a digital simulation routine (and most notably Trotterization) is often needed to implement the simulation \cite{2014_simulation_review,childs2021_trotter,lloyd1996_universal_quantum_simulators}, which will result in a circuit depth overhead. However, an analysis of the quantum simulation routines lies beyond the scope of this work, and we will not make any further assumptions regarding implementation besides the existence of a circuit of the form of Eq.~\ref{eq:ideal_time_evolution}. We remark that Refs.~\cite{trivedi2025stability_mappings, eisert2025stability,Zoller2019_trotter} studied the stability to noise of Trotterization for spin systems, and found that a circuit with system size independent depth $d=\Theta(t^{D+2}/\varepsilon)$ is sufficient to simulate local observables to an accuracy $\varepsilon$ after the time evolution of a local Hamiltonian for a time $t$. As a consequence, constant depth circuits $d=\mathcal{O}(1)$ suffice to simulate constant time evolution $t=\mathcal{O}(1)$, which implies that a constant depth circuit of the form of Eq.~\ref{eq:ideal_time_evolution} can encode the constant time evolution of local Hamiltonians even when Trotterization is needed. Furthermore, we remark that the prospects for quantum advantage in the constant time simulation of local observables were studied in detail in Refs.~\cite{trivedi2025_accuracy_guarantees,trivedi2024analog_stability}.

\begin{figure}[t]
 
    \includegraphics[width=0.6\linewidth]{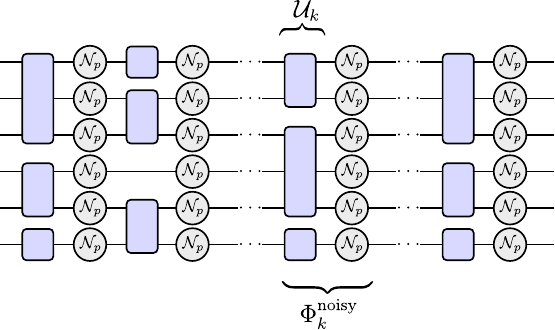}
 
    \caption{\justifying
        Sketch of the circuits considered in Eq.~\ref{eq:noisy_time_evolution}. Each circuit contains $d$ layers of local unitaries, and each layer of unitaries is followed by a layer of single-qubit noise (depicted by gray circle). For simplicity, we have depicted a circuit on $6$ qubits, where the unitaries are $3-$qubit gates. The unitary layer at time step $k$ is denoted by $\mathcal{U}_k$, and the noisy layer by $\Phi_k^{\mathrm{noisy}}=\mathcal{N}_p \circ \mathcal{U}_k$.
    }\label{fig:circuit}
\end{figure}
To model the noise process, we will consider the following noisy evolution
\begin{align}\label{eq:noisy_time_evolution}
\rho_{\mathrm{noisy}}(d)=\Phi_d^{\mathrm{noisy}} \Phi_{d-1}^{\mathrm{noisy}} \dots \Phi_1^{\mathrm{noisy}} \rho(0),
\end{align}
 where $\Phi_\tau^{\mathrm{noisy}}:=\mathcal{N}_p \circ \mathcal{U}_\tau$, with $\mathcal{N}_p$ a layer of Pauli noise channels of strength $p$ as described in Eq.~\ref{eq:Pauli_noise}. Hence, the noisy quantum circuit is modeled by layers of unitary gates followed by layers of single-qubit Pauli noise. This is depicted graphically in Fig. \ref{fig:circuit}. We will consider the measurement of quadratic observables $O$, with the general expression for $O$ provided in Eq.~\ref{eq:expression_Ok_2D}. Then, the ideal expectation value is denoted by $\langle O \rangle_{\mathrm{noiseless}}:=\mathrm{tr}(O\rho_{\mathrm{id}}(t))$, and the noisy expectation value of the observable is $\langle O \rangle_{\mathrm{noisy}}:=\mathrm{tr}(O\rho_{\mathrm{noisy}}(t))$.

Using Eqs.~\ref{eq:ideal_time_evolution} and \ref{eq:noisy_time_evolution}, we will write the error as
\begin{align}\label{eq:Error_noisy_dynamics}
\mathrm{Error}=|\langle O \rangle_{\mathrm{noiseless}}-\langle O \rangle_{\mathrm{noisy}}|  \leq \sum_{k=1}^d \left|\mathrm{tr}(\Phi_d^{\mathrm{id}} \dots \Phi_{k+1}^{\mathrm{id}}(\mathcal{N}_p-\mathrm{id})\Phi_k^{\mathrm{id}}\Phi_{k-1}^{\mathrm{noisy}} \dots \Phi_1^{\mathrm{noisy}}\rho(0) O)\right|=\sum_{k=1}^d |M_k|.
\end{align}

Let us rewrite 
\begin{align}\label{eq:expression_Mk}
M_k:=\mathrm{tr}\left(\left(\mathcal{N}_p-\mathrm{id}\right) \rho^{\mathrm{noisy}}_{(k)} O^{\mathrm{id}}_{(k+1,d)}\right),
\end{align}
where $\rho^{\mathrm{noisy}}_{(k)}:=\Phi_k^{\mathrm{id}}\Phi_{k-1}^{\mathrm{noisy}} \dots \Phi_1^{\mathrm{noisy}} (\rho_0)$ represents the noisy state after $k$ time steps, and 

\begin{align}\label{eq:O_ideal_k}
O^{\mathrm{id}}_{(k+1,d)}:=\left(\left[\Phi_{k+1}^{\mathrm{id}}\right]^{\dagger}\left[\Phi_{k+2}^{\mathrm{id}}\right]^{\dagger} \dots \left[\Phi_{d}^{\mathrm{id}}\right]^{\dagger} \right) O
\end{align}


\noindent represents the Heisenberg evolved operator up to time step $k+1$ under the ideal evolution. Analogously to Proposition \ref{proposition:stability_2D}, the decay of correlations will be the key factor that allows us to show stability. Therefore, we will now provide a Lemma that bounds the decay of correlations in the state after the noisy time evolution. We note that, in general, the correlation function between two operators $O_A$ and $O_B$ is defined as $\langle O_AO_B\rangle-\langle O_A\rangle\langle O_B\rangle$. However, in this section we will often deal with odd operators $O_A$ and $O_B$, meaning that they can be written as a linear combination of an odd product of Majoranas, and therefore do not conserve fermionic parity. Since fermionic physical states have fixed parity due to superselection rules, this implies that $\langle O_A\rangle=\langle O_B\rangle=0$, and the correlation function becomes simply $\langle O_AO_B\rangle$. Particularly, this will be true whenever $O_A$ and $O_B$ are Majorana operators themselves, or are Heisenberg evolved Majoranas with a physical Hamiltonian that preserves fermionic parity. Therefore, in the remainder of the section, we will always assume that $O_A$ and $O_B$ are odd operators, and focus on the correlation function $\langle O_AO_B\rangle$. 

We bound the decay of correlations in the noisy evolved state as follows. 
\begin{lemma}\label{lemma:correlations_interacting}
Consider an initial state  $\rho_0$ which satisfies the clustering of correlations
\begin{equation}
\left| \langle O_AO_B\rangle_{\rho_0}\right| \leq \|O_A\| \|O_B\|f(d(A,B)), 
\end{equation}
with $O_A$ and $O_B$ odd operators supported in disjoint subregions $A$ and $B$ respectively, and $f(d(A,B))$ a function that decays with the distance between $A$ and $B$.
Assume the noiseless evolution in Eq.~\ref{eq:ideal_time_evolution} is composed of  local gates. Then, the noisy time evolved state $\rho_{(k)}^{\mathrm{noisy}}:=\Phi_k^{\mathrm{id}}\Phi_k^{\mathrm{noisy}}\dots\Phi_1^{\mathrm{noisy}}\rho_0$ satisfies
\begin{equation}
\left| \langle O_AO_B\rangle_{\rho^{\mathrm{noisy}}_{(k)}}\right| \leq  \|O_A\| \|O_B\| f(d(A,B)-2vk),
\end{equation}
for some constant $v=\mathcal{O}(1)$, which is the Lieb-Robinson velocity.
\end{lemma}

The proof for Lemma \ref{lemma:correlations_interacting} is adapted from the usual Lieb-Robinson bounds, with the additional complexity that the noise channel is not local in the fermions \footnote{Precisely, the noise does not factorize in the fermions, and in general $\mathcal{N}_p(O_AO_B) \neq \mathcal{N}_p(O_A) \mathcal{N}_p(O_B)$ for disjoint operators $O_A$ and $O_B$.}. A detailed proof can be found in Appendix \ref{appendix:correlations}. Lemma \ref{lemma:correlations_interacting} thus allows us to bound the correlations in the noisy state after $k$ time steps, and implies that the correlations follow the usual Lieb-Robinson light cone. This bound on the decay of correlations will allow us to show stability following the same principles discussed in section \ref{section:stability_encodings}.

Note that we have not yet assumed a free fermionic evolution in the Lemma, and the correlation decay holds for interacting systems as well. In subsection \ref{subsection:noisy_evolution_free} we will consider Gaussian circuits which give rise to a free fermions evolution, while in subsection \ref{subsection:noisy_evolution_interacting} we will analyze more general interacting circuits. Similarly to previous works \cite{trivedi2024analog_stability}, we are able to provide tighter bounds for the free fermion case, which motivates the presentation of the two results separately. 

\subsection{Free fermions}\label{subsection:noisy_evolution_free}

In this subsection, we will assume that each layer of unitaries $ \mathcal{U}_\tau$ in Eq.~\ref{eq:noisy_time_evolution} is composed of geometrically local fermionic Gaussian gates. In order to show stability, we will rely on Lemma \ref{lemma:correlations_interacting}, which shows that a constant time evolution cannot generate long-range correlations due to Lieb-Robinson constraints. Hence, if the initial state exhibits sufficient decay of correlations, the state will exhibit this decay at all times, which from section \ref{section:stability_encodings} we know to be sufficient to provide stability of measurements. This intuition is formalized in the following Proposition.

\begin{proposition}
Consider the noisy evolution of Eq.~\ref{eq:noisy_time_evolution} under local unitary gates, with initial state $\rho_0$, and circuit depth $d=\mathcal{O}(1)$. Assume that the correlations in the initial state decay as
\begin{equation}
\left| \Gamma_{\mathbf{r},\mathbf{r}'}^{\alpha,\beta}(\rho_0) \right| \leq \frac{K}{d(\mathbf{r},\mathbf{r'})^{\mu}}, \forall (\mathbf{r} \neq \mathbf{r'},\alpha,\beta).
\end{equation}

\noindent Consider the task of measuring a quadratic observable $O$. Then, the task is stable to noise when $\mu>D$, with 
\begin{equation}
\mathrm{Error}=|\langle O\rangle_{\mathrm{noiseless}}-\langle O\rangle_{\mathrm{noisy}}| \leq \mathcal{O}\left(f(p)d^{D+2}\right),
\end{equation}
where
\begin{align}
f(p) = 
\begin{cases} 
\mathcal{O}\left( p^{\mu-D} \right) & \text{if } D < \mu < D+1, \\ 
\mathcal{O}\left( p \log \frac{1}{p} \right) & \text{if } \mu = D+1, \\ 
\mathcal{O}(p) & \text{if } \mu > D+1 ,
\end{cases}
\end{align}
\noindent ensuring $\lim_{p \to  0} f(p)=0$.

\end{proposition}

\begin{proof}
We will first use Eqs.~\ref{eq:Error_noisy_dynamics} and \ref{eq:expression_Mk} to write the error as 
\begin{align}\label{eq:error_2}
|\langle O\rangle_{\mathrm{noiseless}}-\langle O\rangle_{\mathrm{noisy}}| \leq \sum_{k=1}^d |M_k|,
\end{align}
and we will bound the contribution of each $|M_k|$ separately. Recall the expression from Eq.~\ref{eq:expression_Mk}:
\begin{align}
M_k:=\mathrm{tr}\left(\left(\mathcal{N}_p - \mathrm{id}\right) \rho^{\mathrm{noisy}}_{(k)} O^{\mathrm{id}}_{(k+1,d)}\right).
\end{align}

\noindent Note that, in this expression, $O^{\mathrm{id}}_{(k+1,d)}$ is obtained by evolving $O$ in the Heisenberg picture via the ideal evolution,
\begin{align}
O^{\mathrm{id}}_{(k+1,d)}:=\left(\left[\Phi_{k+1}^{\mathrm{id}}\right]^{\dagger}\left[\Phi_{k+2}^{\mathrm{id}}\right]^{\dagger} \dots \left[\Phi_{d}^{\mathrm{id}}\right]^{\dagger} \right) O.
\end{align}

\noindent Since $O$ is quadratic and the ideal evolution is generated by a Hamiltonian $H$ which is also quadratic, it follows that $O^{\mathrm{id}}_{(k+1,d)}$ is also quadratic. As a consequence, one can always write $O^{\mathrm{id}}_{(k+1,d)}$ using the general expression in Eq. \ref{eq:expression_Ok_2D},

\begin{align}
O^{\mathrm{id}}_{(k+1,d)}=i\sum_{\mathbf{r},\mathbf{r}'}\sum_{\alpha,\beta} O_{\mathbf{r},\mathbf{r}'}^{\alpha,\beta} \gamma_{\mathbf{r}}^{\alpha}\gamma_{\mathbf{r}'}^{\beta}.
\end{align}

\noindent We can then write

\begin{align}\label{eq:expression_MK_expanded}
M_k=\mathrm{tr}\left(\left(\mathcal{N}_p - \mathrm{id}\right) \rho^{\mathrm{noisy}}_{(k)} O^{\mathrm{id}}_{(k+1,d)}\right)=\sum_{\mathbf{r},\mathbf{r}'}\sum_{\alpha,\beta} O_{\mathbf{r},\mathbf{r}'}^{\alpha,\beta} \left\langle \left(\mathcal{N}_p-\mathrm{id}\right)\gamma_{\mathbf{r}}^{\alpha}\gamma_{\mathbf{r}'}^{\beta} \right\rangle_{\rho_{(k)}^{\mathrm{noisy}}}=\sum_{\mathbf{r},\mathbf{r}'}\sum_{\alpha,\beta} O_{\mathbf{r},\mathbf{r}'}^{\alpha,\beta}\left(\lambda_{\mathbf{r},\mathbf{r}'}^{\alpha,\beta}-1\right)\Gamma_{\mathbf{r},\mathbf{r}'}^{\alpha,\beta}\left(\rho_{(k)}^{\mathrm{noisy}}\right),
\end{align}

\noindent where we have used $\mathcal{N}_p \left(\gamma_{\mathbf{r}}^{\alpha}\gamma_{\mathbf{r}'}^{\beta}\right)=\lambda_{\mathbf{r},\mathbf{r'}}^{\alpha,\beta}\left(\gamma_{\mathbf{r}}^{\alpha}\gamma_{\mathbf{r}'}^{\beta}\right)$ coming from Eq.~\ref{eq:dep_noise_1}. Furthermore, note that $\rho^{\mathrm{noisy}}_{(k)}$ is obtained by applying $k$ layers of the noisy evolution, $\rho^{\mathrm{noisy}}_{(k)}:=\Phi_{k}^{\mathrm{id}}\Phi_{k-1}^{\mathrm{noisy}} \dots \Phi_1^{\mathrm{noisy}} (\rho_0)$. Hence, by direct application of Lemma \ref{lemma:correlations_interacting}, we can bound the correlations as

\begin{align}\label{eq:corr_matrix_decay}
\left| \Gamma_{\mathbf{r},\mathbf{r}'}^{\alpha,\beta}\left(\rho^{\mathrm{noisy}}_{(k)} \right) \right| \leq\frac{K}{\left[d(\mathbf{r},\mathbf{r'})-2vk\right]^{\mu}},
\end{align}

\noindent for some $v=\mathcal{O}(1)$. 
Using Eq.~\ref{eq:expression_MK_expanded}, we write

\begin{align}
|M_k| =   \left|\sum_{\alpha,\beta} \sum_{\mathbf{r}, \mathbf{r'}}O^{\alpha,\beta}_{\mathbf{r},\mathbf{r}'} \left[1-\lambda_{\mathbf{r},\mathbf{r'}}^{\alpha,\beta}\right]\Gamma^{\alpha,\beta}_{\mathbf{r},\mathbf{r}'}\right|=  \left|\sum_{\alpha,\beta} \sum_{\mathbf{r}, \mathbf{r'}}\widetilde{O}_{(\mathbf{r},\alpha),(\mathbf{r}',\beta)}F_{(\mathbf{r},\alpha),(\mathbf{r}',\beta)}\right| =  \left| \mathrm{tr}(\tilde{O}F)\right| \leq  \|\widetilde{O}\|_1 \|F\|,
\end{align}

\noindent where $\widetilde{O}$ and $F$ are matrices of size $2N \times 2N$ with indices $(\mathbf{r},\alpha)$ and $(\mathbf{r}',\beta)$, defined as

\begin{align}
\widetilde{O}_{(\mathbf{r},\alpha),(\mathbf{r}',\beta)} := O^{\alpha,\beta}_{\mathbf{r},\mathbf{r'}} \quad ; \quad F_{(\mathbf{r},\alpha),(\mathbf{r}',\beta)}:=\left[1-\lambda_{\mathbf{r},\mathbf{r'}}^{\alpha,\beta}\right]\Gamma^{\alpha,\beta}_{\mathbf{r},\mathbf{r}'},
\end{align}
\noindent respectively. Note that $\|\widetilde{O}\|_1 \leq 1$ due to the normalization chosen in Eq.~\ref{eq:expression_Ok_2D}, and this is true at all times since it is a unitarily invariant norm. Furthermore, $F$ is Hermitian, and from Eq.~\ref{eq:bound_eigenvalues} we know that $1-\lambda_{\mathbf{r,\mathbf{r'}}}^{\alpha,\beta} \leq 1-(1-3p/2)^{\varphi(\mathbf{r},\mathbf{r}')}$. Denoting $r=1-3p/2$, this allows us to bound its operator norm as

\begin{align}\label{eq:sum_F_noninteracting}
\|F\| \leq \max_{(\mathbf{r}',\beta)}\sum_{(\mathbf{r},\alpha)}|F_{(\mathbf{r},\alpha),(\mathbf{r'},\beta)}|  \leq  \max_{(\mathbf{r}',\beta)}\sum_{(\mathbf{r},\alpha)} \left(1-r^{d(\mathbf{r},\mathbf{r}')+\varphi_0}\right) \left| \Gamma^{\alpha,\beta}_{\mathbf{r},\mathbf{r}'}\left(\rho_{(k)}^{\mathrm{noisy}}\right)\right|.
\end{align}

For correlation matrices that satisfy the correlation decay in Eq.~\ref{eq:corr_matrix_decay}, the sum in Eq.~\ref{eq:sum_F_noninteracting}, is bounded in Lemma \ref{lemma:bound_S} of Appendix \ref{appendix:sum}, which yields $\|F\| \leq \mathcal{O}\left(d^{D+1}\right) f(r)$, with the behavior of $f(r)$ in the small error limit bounded as 

\begin{align}
f(r) \leq 
\begin{cases} 
\mathcal{O}\left( (1-r)^{\mu-D}\right) & \text{if } D < \mu < D+1, \\ 
\mathcal{O}\left( (1-r) \log (1/(1-r)\right) & \text{if } \mu = D+1, \\ 
\mathcal{O}(1-r) & \text{if } \mu > D+1 .
\end{cases}
\end{align}
Therefore, the total error is bounded as 

\begin{align}
\sum_k |M_k| \leq \sum_{k=1}^d \mathcal{O}(d^{D+1})f(p)= f(p) \mathcal{O}\left( d^{D+2} \right).
\end{align}

\noindent which completes the proof.
\end{proof}

We have therefore obtained that in the free fermion case, for circuits of depth $d$ the error scales as $f(p)\mathcal{O}(d^{D+2})$, with $f(p)=\mathcal{O}(p)$ when $\mu>D+1$. By contrast, Ref.~\cite{trivedi2024analog_stability} showed that the constant-time dynamics under a Gaussian, local Hamiltonian in a fermionic analog quantum simulator is stable with error scaling as $\mathcal{O}(pt)$, independent of the dimensionality. This difference in scaling arises from the overhead of the fermion-to-qubit mapping, which maps Majorana bilinears to (potentially) long Pauli strings which are more sensitive to noise, hence giving an additional $d^{D+1}$ factor. We also remark that the analysis in Ref.~\cite{trivedi2024analog_stability} includes both coherent and incoherent (local) noise, while we have restricted ourselves to incoherent Pauli noise, which becomes effectively non-local due to the fermionic encoding.

\subsection{Interacting systems}\label{subsection:noisy_evolution_interacting}

We will now turn to interacting systems, and we will assume that each layer of unitaries $ \mathcal{U}_k$ in Eq.~\ref{eq:ideal_time_evolution} is composed of geometrically local gates which are not necessarily Gaussian. In an analogous way to the Gaussian case, a decay of correlations in the initial state is sufficient to prove stability bounds for the interacting case.

\begin{proposition}\label{proposition:time_evolution_interacting}
Consider the noisy evolution of Eq.~\ref{eq:noisy_time_evolution} under a local Hamiltonian $H$, with initial state $\rho_0$. Assume that the correlations in the initial state decay as

\begin{equation}
\left| \langle O_AO_B\rangle_{\rho_0}\right| \leq \|O_A\| \|O_B\| \frac{K}{d(A,B)^{\mu}},
\end{equation}

\noindent where $O_A$ and $O_B$ are odd operators supported on disjoint regions $A$ and $B$, respectively. Consider the task of measuring a quadratic observable $O$. Then, the task is stable to noise when $\mu>D$, with 

\begin{equation}
\mathrm{Error}=|\langle O\rangle_{\mathrm{noiseless}}-\langle O\rangle_{\mathrm{noisy}}| \leq f(p)\mathcal{O}\left(d^{2D+1}\right) ,
\end{equation}

\noindent where

\begin{align}
f(p) = 
\begin{cases} 
\mathcal{O}\left( p^{\mu-D} \right) & \text{if } D < \mu < D+1, \\ 
\mathcal{O}\left( p \log \frac{1}{p} \right) & \text{if } \mu = D+1, \\ 
\mathcal{O}(p) & \text{if } \mu > D+1 ,
\end{cases}
\end{align}
\noindent ensuring $\lim_{p \to  0} f(p)=0$.
\end{proposition}

\begin{proof}
	We first use Eqs.~\ref{eq:Error_noisy_dynamics} and \ref{eq:expression_Mk} to write the error as 

    \begin{align}
    |\langle O\rangle_{\mathrm{noiseless}}-\langle O\rangle_{\mathrm{noisy}}| \leq \sum_{k=1}^d |M_k|.
    \end{align}
    
    We start by bounding the term $|M_k|$ in the error sum. As in Eq.~\ref{eq:expression_Mk}, we have
	\begin{align}
		M_k := \mathrm{tr}\left({\rho}_{k}^{\mathrm{noisy}} (\mathcal{N}_p-\mathrm{id})O^{\mathrm{id}}_{(k+1,d)}\right),
	\end{align}
	where $O^{\mathrm{id}}_{(k+1,d)}$ is the ideal Heisenberg evolution of the observable $O$ from time $d$ back to $k+1$ from Eq.~\ref{eq:O_ideal_k}. Note that, while the observable $O$ is quadratic,
	the operator $O^{\mathrm{id}}_{(k+1,d)}$ is no longer quadratic after the interacting evolution. However, due to the locality of the circuit, the operator $\gamma_{\mathbf{r}}^{\alpha}$ evolves into an operator supported on a ball $\mathcal{B}_{\mathbf{r}}$ of radius $v(d-k)$ centered at $\mathbf{r}$, with $v=\mathcal{O}(1)$ the Lieb-Robinson velocity. Let us write

    \begin{align}\label{eq:expression_Mk_Arr}
    M_k=\sum_{\alpha,\beta} \sum_{\mathbf{r},\mathbf{r}'}O_{\mathbf{r},\mathbf{r}'}^{\alpha,\beta}\left \langle \left(\mathcal{N}_p-\mathrm{id}\right)A_{\mathbf{r},\mathbf{r}'}^{\alpha,\beta} \right \rangle_{\rho^{\mathrm{noisy}}_{(k)}},
    \end{align}

    \noindent where we have denoted the Heisenberg evolved Majorana bilinears at time $d-k$ as
    
    \begin{align}
        A_{\mathbf{r},\mathbf{r}'}^{\alpha,\beta}:= \left[\Phi_{k+1}^{\mathrm{id}}\right]^{\dagger} \dots \left[\Phi_{d}^{\mathrm{id}}\right]^{\dagger}(\gamma_{\mathbf{r}}^\alpha \gamma_{\mathbf{r}'}^\beta)= (\gamma_{\mathbf{r}}^\alpha \gamma_{\mathbf{r}'}^\beta)(d-k).
    \end{align}

    Note that the Pauli noise channel can always be written as $\mathcal{N}_p(\cdot)=\sum_S q_S S(\cdot) S$, where the sum is over all the possible Pauli strings $S$ with probabilities $q_S$ such that $\sum_S q_S=1$. We can then write

    \begin{align}
        (\mathcal{N}_p - \mathrm{id})A_{\mathbf{r},\mathbf{r}'}^{\alpha,\beta} =\sum_{S} q_S  S \left( A_{\mathbf{r},\mathbf{r}'}^{\alpha,\beta}\right) S -A_{\mathbf{r},\mathbf {r'}}^{\alpha,\beta}.
    \end{align}

    Let us now analyze the action of each Pauli string $S$. Note that, due to the strict locality of the unitary layers, one can write  $A_{\mathbf{r},\mathbf{r}'}^{\alpha,\beta}=\gamma_{\mathbf{r}}^{\alpha}(d-k) \gamma_{\mathbf{r'}}^{\beta}(d-k)$, with $\supp\left(\gamma_{\mathbf{r}}^{\alpha}(d-k) \right)\in \mathcal{B}_{\mathbf{r}}$ and $\supp\left(\gamma_{\mathbf{r'}}^{\beta}(d-k) \right)\in \mathcal{B}_{\mathbf{r'}}$. That is, $A_{\mathbf{r},\mathbf{r}'}^{\alpha,\beta}$ is structured as a product of two local operators, one supported in $\mathcal{B}_{\mathbf{r}}$ and the other supported in $\mathcal{B}_{\mathbf{r}'}$. As a consequence, in the qubit picture, $A_{\mathbf{r},\mathbf{r}'}^{\alpha,\beta}$  is mapped to a qubit operator supported in a region $\mathcal{R}$ of size $|\mathcal{R}|=c_1(d-k)^D+c_2 (d-k)^{D-1} d(\mathbf{r},\mathbf{r}')$, with $c_1,c_2=\mathcal{O}(1)$, where the first term comes from the size of the regions $\mathcal{B}_{\mathbf{r}}$ and $\mathcal{B}_{\mathbf{r}'}$ in which the fermionic operator is supported, and the term $c_2 (d-k)^{D-1} d(\mathbf{r},\mathbf{r}')$ bounds the size of the region of all the possible Pauli strings connecting $\mathcal{B}_{\mathbf{r}}$ and $\mathcal{B}_{\mathbf{r}'}$, where they are assumed to follow the shortest path according to Assumption \ref{assumption:encodings}. Note that all the Pauli strings $S$ which have no support in $\mathcal{R}$ will commute with $A_{\mathbf{r},\mathbf{r'}}^{\alpha,\beta}$, and therefore will act trivially. We can then write

    \begin{align}\label{eq:Np_A}
        (\mathcal{N}_p - \mathrm{id})A_{\mathbf{r},\mathbf{r}'}^{\alpha,\beta} =\sum_{\supp(S) \in \mathcal{R}} q_S  S \left( A_{\mathbf{r},\mathbf{r}'}^{\alpha,\beta}\right) S+ (q_I-1) A_{\mathbf{r},\mathbf{r}'}^{\alpha,\beta}.
    \end{align}

    \noindent Here, $q_I$ includes all the Pauli strings that are not supported in $\mathcal{R}$. Intuitively, the probability $q_I$ is the probability that no error occurs inside $\mathcal{R}$, and can be bounded using Eq. \ref{eq:Pauli_noise} as

    \begin{align}\label{eq:bound_pI}
    q_I=\left(1-\frac{3}{4}p\right)^{|\mathcal{R}|}= 1-r^{c_1(d-k)^D+c_2 (d-k)^{D-1} d(\mathbf{r},\mathbf{r}')},
    \end{align}

    \noindent where we denote $r:=1-3p/4$. This also implies that $\sum_{\supp (S) \in \mathcal{R}} q_S =1-q_I = \left(1-r^{|\mathcal{R}|} \right)$. We can then use Eq. \ref{eq:Np_A} to bound
    \begin{align}\label{eq:bound_corr_decay2}
        \left| \left\langle (\mathcal{N}_p - \mathrm{id})A_{\mathbf{r},\mathbf{r}'}^{\alpha,\beta} \right \rangle_{\rho^{\mathrm{noisy}}_{(k)}}\right|  &= \left|\sum_{\supp(S) \in \mathcal{R}} q_S   \left \langle S \left( A_{\mathbf{r},\mathbf{r}'}^{\alpha,\beta}\right) S \right \rangle_{\rho^{\mathrm{noisy}}_{(k)}}+ (q_I-1) \left \langle A_{\mathbf{r},\mathbf{r}'}^{\alpha,\beta} \right \rangle_{\rho^{\mathrm{noisy}}_{(k)}}\right| \nonumber \\ &\leq \left(\sum_{\supp(S) \in \mathcal{R}} q_S  +  |q_I-1| \right) \max_S \left|\left \langle SA_{\mathbf{r},\mathbf{r}'}^{\alpha,\beta} S\right \rangle_{\rho^{\mathrm{noisy}}_{(k)}} \right| \leq 2\left( 1-r^{|\mathcal{R}|}\right) \max_S \left|\left \langle SA_{\mathbf{r},\mathbf{r}'}^{\alpha,\beta} S\right \rangle_{\rho^{\mathrm{noisy}}_{(k)}} \right|.
    \end{align}

    \noindent Now, note that $SA_{\mathbf{r},\mathbf{r}'}^{\alpha,\beta} S  = \widetilde{O}_{\mathcal{B}_{\mathbf{r}}}\widetilde{O}_{\mathcal{B}_{\mathbf{r'}}}$, where $\widetilde{O}_{\mathcal{B}_{\mathbf{r}}}:=S \gamma_{\mathbf{r}}^{\alpha}(d-k)S
    $ is supported in $\mathcal{B}_{\mathbf{r}}$. This holds because $\supp\left(\gamma_{\mathbf{r}}^{\alpha}(d-k)\right) \in \mathcal{B}_{\mathbf{r}}$, and since the Pauli string $S$ is mapped to a Majorana string, it cannot change the locality of the operator. This also implies that $\widetilde{O}_{\mathcal{B}_{\mathbf{r}'}}=S \gamma_{\mathbf{r'}}^{\beta}(d-k)S
    $ is supported in $\mathcal{B}_{\mathbf{r}'}$ for any $S$. Note also that $d(\mathcal{B}_{\mathbf{r}},\mathcal{B}_{\mathbf{r'}})\geq d(\mathbf{r},\mathbf{r}')-2v(d-k)$, for some constant $v=\mathcal{O}(1)$. Hence, using the correlation decay of Lemma \ref{lemma:correlations_interacting} we can write

    \begin{align}\label{eq:bound_corr_decay}
        \max_S \left|\left \langle SA_{\mathbf{r},\mathbf{r}'}^{\alpha,\beta} S\right \rangle_{\rho^{\mathrm{noisy}}_{(k)}} \right| \leq \frac{K}{\left[d(\mathbf{r}, \mathbf{r'})-2vd\right]^{\mu}} \quad  \mathrm{if} \; d(\mathbf{r},\mathbf{r'}) > 2vd,
    \end{align}

    \noindent and 
    \begin{equation}\max_S \left|\left \langle SA_{\mathbf{r},\mathbf{r}'}^{\alpha,\beta} S\right \rangle_{\rho^{\mathrm{noisy}}_{(k)}} \right| \leq 1 \quad \text{if } d(\mathbf{r},\mathbf{r}') \leq 2vd,
    \end{equation} since $\|SA_{\mathbf{r},\mathbf{r}'}^{\alpha,\beta} S\|=1$.  
    
    Then, using Eqs.~\ref{eq:bound_corr_decay2} and \ref{eq:bound_corr_decay}, as well as the expression for $|\mathcal{R}|$, yields

    \begin{align}
    \left| \left\langle (\mathcal{N}_p - \mathrm{id})A_{\mathbf{r},\mathbf{r}'}^{\alpha,\beta} \right \rangle_{\rho^{\mathrm{noisy}}_{(k)}}\right|  \leq  2K\frac{\left(1-r^{c_1d(\mathbf{r},\mathbf{r'})(d-k)^{D-1}+c_2 (d-k)^D} \right)}{\left[d(\mathbf{r},\mathbf{r'})-2vd\right]^{\mu}}.
    \end{align}

    Let us now use Eq. \ref{eq:expression_Mk_Arr} to bound the error corresponding to $M_k$ as 
    
    \begin{align}\label{eq:norm_Mk_interacting}
        |M_k|=\left|\sum_{\alpha,\beta} \sum_{\mathbf{r},\mathbf{r}'}O_{\mathbf{r},\mathbf{r}'}^{\alpha,\beta}\left \langle \left(\mathcal{N}_p-\mathrm{id}\right)A_{\mathbf{r},\mathbf{r}'}^{\alpha,\beta} \right \rangle_{\rho^{\mathrm{noisy}}_{(k)}}\right|=\left|\mathrm{tr}\left(\widetilde{O}F\right)\right| \leq \|\widetilde{O}\|_1 \|F\|,
    \end{align}

    \noindent where $\widetilde{O}$ and $F$ represent matrices of size $2N \times 2N$ with indices $(\mathbf{r},\alpha)$ and $(\mathbf{r}',\beta)$, defined as

    \begin{align}
    \widetilde{O}_{(\mathbf{r},\alpha),(\mathbf{r}',\beta)} := O^{\alpha,\beta}_{\mathbf{r},\mathbf{r'}} \quad ; \quad F_{(\mathbf{r},\alpha),(\mathbf{r}',\beta)}:=\left \langle \left(\mathcal{N}_p-\mathrm{id}\right)A_{\mathbf{r},\mathbf{r}'}^{\alpha,\beta} \right \rangle_{\rho^{\mathrm{noisy}}_{(k)}}.
    \end{align}

    \noindent Note that $\|\widetilde{O}\|_1\leq 1$ due to the chosen normalization, which implies that $|M_k| \leq \|F\|$. Furthermore, since $F$ is hermitian, its operator norm $\|F\|$ can be bounded using Eq.~\ref{eq:norm_inequality} as

    \begin{align}
        \|F\| \leq \max_{(\mathbf{r'},\beta)} \sum_{(\mathbf{r},\alpha)} \left| F_{(\mathbf{r},\alpha),(\mathbf{r'},\beta)}\right| \leq 2K \max_{(\mathbf{r}',\beta)}\sum_{(\mathbf{r},\alpha)}\frac{\left(1-r^{c_1d(\mathbf{r},\mathbf{r'})(d-k)^{D-1}+c_2 (d-k)^D} \right)}{\left[d(\mathbf{r},\mathbf{r'})-2vd\right]^{\mu}}.
    \end{align}
    
    This sum is rigorously bounded in Corollary \ref{corollary:sum_2} of Appendix \ref{appendix:sum}, which yields 

    \begin{align}\label{eq:norm:F_interacting}
        \|F\|  \leq 2K \max_{(\mathbf{r}',\beta)}\sum_{(\mathbf{r},\alpha)}\frac{\left(1-r^{c_1d(\mathbf{r},\mathbf{r'})(d-k)^{D-1}+c_2 (d-k)^D} \right)}{\left[d(\mathbf{r},\mathbf{r'})-2vd\right]^{\mu}}=\mathcal{O}\left(d^{2D}\right)f(r),
    \end{align}

    \noindent where 
    \begin{align}
    f(r)=\begin{cases}
    \mathcal{O}\left((1-r)^{\mu-D}\right) & \text{if } D < \mu < D+1,\\
    {\mathcal{O}}\left((1-r) \left[\log (1/(1+r))\right]\right) & \text{if } \mu = D+1,\\
    {\mathcal{O}}\left(1-r\right) & \text{if } \mu > D+1.
    \end{cases}
    \end{align}

    Hence, putting together Eqs.~\ref{eq:norm_Mk_interacting} and \ref{eq:norm:F_interacting} we obtain $|M_k| \leq \mathcal{O} \left(d^{2D} \right) f(r)$. Finally, summing over $k=1$ to $d$ we get
	\begin{align}
		\mathrm{Error} \leq \sum_{k=1}^d |M_k| \leq f(p) \mathcal{O}\left(d^{2D+1}\right).
	\end{align}
	\noindent This completes the proof.
\end{proof}

We have therefore obtained that even in the interacting case the noisy constant-time evolution is stable to noise, with error scaling as $f(p) \mathcal{O}\left(d^{2D+1}\right)$. This scaling is steeper than the bound $f(p) \mathcal{O}\left(d^{D+2}\right)$ obtained in the free fermion case, which is to be expected since the presence of interaction terms imply that the time evolved observable is no longer quadratic, and thus can be more susceptible to noise. Furthermore, for quantum devices which can simulate fermions natively and do not need to use a fermionic encoding, in the presence of local noise the error can be bounded as $\mathcal{O}\left(pd^{D+1}\right)$ \cite{trivedi2024analog_stability,eisert2025stability}, which again represents a milder scaling as compared to our result. Analogously to the free fermion case, this difference in scaling is explained by the overhead introduced by the fermion-to-qubit mapping, which adds an additional factor of $\mathcal{O}\left(d^D\right)$ to our bound.

\section{Conclusions}

In this work, we have analyzed the stability to noise of fermion-to-qubit mappings. We first investigated the measurement of quadratic observables to probe the stability properties of several fermionic encodings. For a family of local encodings, we provided a rigorous stability criterion: For $D$-dimensional systems, sufficient decay of correlations (specifically  $|\langle c_{\mathbf{r}}^{\dagger}c_{\mathbf{r'}} \rangle|\leq 1/d(\mathbf{r},\mathbf{r}')^{\mu}$, with $\mu>D$) ensures stability against general Pauli noise. Our results imply that there is a sharp transition in the stability to noise of local fermionic encodings as a function of the correlation decay. Conversely, we demonstrated that the Jordan-Wigner encoding with snake ordering in two dimensions, or even quasi-local encodings such as the Bravyi-Kitaev transform are asymptotically fragile to noise.

Furthermore, we extended the analysis to a specific class of strongly correlated states that do not satisfy the correlation decay condition above, and  identified a direct link between the regularity of the momentum occupation function and noise stability. For states with a Fermi surface, we showed that the measurement of the momentum occupation function $n_k=a_k^{\dagger}a_k$ is stable everywhere except in the immediate vicinity of the Fermi surface, where the discontinuity leads to a fragility that scales with the size of the system. This suggests that noise sensitivity itself could be used as a diagnostic tool to map out Fermi surfaces in strongly correlated systems. More generally, for arbitrary states with definite particle number, the measurement of quadratic observables is controlled by the smoothness of the momentum occupation $n_k$: if $n_k$ is a Lipschitz continuous function, then stability is guaranteed.

Finally, we used the stability results above to analyze the noise resilience of constant depth quantum circuits in the presence of Pauli noise, which serves as a model for computational tasks such as noisy state preparation or noisy digital quantum simulation. We proved that the simulation of quadratic observables after constant-time evolution remains stable for constant depth $d=\mathcal{O}(1)$ circuits for both free and interacting fermions, with error scalings of $f(p)\mathcal{O}(d^{D+2})$ and $f(p)\mathcal{O}\left(d^{2D+1}\right)$ respectively, where $f(p)=\mathcal{O}(p)$ for correlations decaying with $\mu>D+1$, and $f(p)=\mathcal{O}(p^{\mu-D})$ for correlation decay $D<\mu<D+1$. While remaining system-size independent, these scalings reveal a steeper scaling as compared to the native implementation of fermionic dynamics in fermionic simulators \cite{trivedi2024analog_stability}, which is a consequence of the overhead from the encoding. Nonetheless, this result rigorously shows the stability of simulating the constant time dynamics of encoded fermionic systems in the presence of Pauli noise, provided that the correlations in the initial state decay sufficiently fast. We leave for future work further generalizations of our results, such as incorporating more realistic noise models (which could include coherent noise or non-Markovian noise, among others).

In summary, our results show that correlation decay protects local fermionic encodings against noise, whereas non-local encodings remain asymptotically fragile. These findings strengthen the available evidence for the feasibility of fermionic simulations on near-term quantum hardware, and offer a practical design principle: For states exhibiting correlation decay, encoding fermionic modes in real space might provide a more robust approach than encoding in momentum space. Since this is true even if the correlation decay is merely algebraic instead of exponential, it covers a wide range of physically relevant states, including thermal states, ground states of gapped Hamiltonians, and a broad class of gapless, interacting systems. This hints that moving away from demonstrations based on random circuits to more structured problems where the state considered is low energy could be beneficial from the point of view of resilience to noise.  Furthermore, this analysis serves as a guiding principle for the development of algorithms involving fermions on digital quantum computers based on qubits. 


\begin{acknowledgements}
We thank Joel Mills and the rest of the Phasecraft team for helpful discussions.
\end{acknowledgements}

\bibliography{references}

\begin{appendix}

\section{Derivation of the sum bound}\label{appendix:sum}


In this section we will provide a bound that is used in the main text. Precisely, we will bound the sum $S$ given by
\begin{align}
S := \max_{\mathbf{r'}} \sum_{\mathbf{r}} \left[1-r^{K_1d(\mathbf{r},\mathbf{r}')+K_2}\right]g(\mathbf{r},\mathbf{r'}),
\end{align}

\noindent with the function $g(\mathbf{r},\mathbf{r'})$ fulfilling
\begin{align}\label{eq:bound_g}
|g(\mathbf{r},\mathbf{r'}) |\leq  
\begin{cases}
1  \quad  \mathrm{if} \quad d(\mathbf{r},\mathbf{r'}) \leq d_0,\\
\frac{K}{\left[d(\mathbf{r},\mathbf{r'})-d_0\right]^{\mu}}   \quad \mathrm{if} \quad d(\mathbf{r},\mathbf{r'}) > d_0.
\end{cases}
\end{align}

\noindent We will provide a bound as function of the parameters $r,K_1,K_2,d_0$. We assume $K=\mathcal{O}(1)$. To do this, we will split the sum as $S=S_1+S_2$, with





\begin{align}
    S_1 := \max_{\mathbf{r'}} \sum_{d(\mathbf{r},\mathbf{r'}) \leq d_0} \left[1-r^{K_1d(\mathbf{r},\mathbf{r}')+K_2}\right]g(\mathbf{r},\mathbf{r'}) \quad ; \quad S_2:= \max_{\mathbf{r'}} \sum_{d(\mathbf{r},\mathbf{r'}) > d_0} \left[1-r^{K_1d(\mathbf{r},\mathbf{r}')+K_2}\right]g(\mathbf{r},\mathbf{r'}).
\end{align}

\noindent We will first bound the sum $S_1$ in the following Lemma

\begin{lemma}\label{lemma:S1}
Consider the sum given by
\begin{align}
    S_1 := \max_{\mathbf{r'}} \sum_{d(\mathbf{r},\mathbf{r'}) \leq d_0} \left[1-r^{K_1d(\mathbf{r},\mathbf{r}')+K_2}\right]g(\mathbf{r},\mathbf{r'}),
\end{align}

\noindent where the function $g(\mathbf{r},\mathbf{r}')$ fulfills the condition in Eq.~\ref{eq:bound_g}. Then, 

\begin{align}
|S_1|=(1-r) (K_1d_0+K_2)\left(1+\mathcal{O} \left(d_0^D\right) \right).
\end{align}
\end{lemma}

\begin{proof}

We write 
\begin{align}
    S_1= \max_{\mathbf{r'}} \sum_{d(\mathbf{r},\mathbf{r'}) \leq d_0} \left[1-r^{K_1d(\mathbf{r},\mathbf{r}')+K_2}\right]g(\mathbf{r},\mathbf{r'})\leq \sum_{d(\mathbf{s},\mathbf{0}) \leq d_0} {1-r^{K_1 d(\mathbf{s},\mathbf{0})+K_2}},
\end{align}

\noindent where we have used the trivial bound  $|g(\mathbf{r},\mathbf{r}')| \leq 1$ from Eq.~\ref{eq:bound_g}, and we have used translation invariance to sum over $\mathbf{s}$. We will now use Bernoulli's inequality $1-r^{K_1 d(\mathbf{s},\mathbf{0})+K_2} \leq (K_1 d(\mathbf{s},\mathbf{0})+K_2) (1-r)$ to write 

\begin{align}
    S_1\leq  \sum_{d(\mathbf{s},\mathbf{0}) \leq d_0} {1-r^{K_1 d(\mathbf{s},\mathbf{0})+K_2}} \leq (1-r) \sum_{d(\mathbf{s,0}) \leq d_0}(K_1 d(\mathbf{s},\mathbf{0})+K_2) \leq (1-r) (K_1d_0+K_2) \sum_{d(\mathbf{s,0}) \leq d_0}1.
\end{align}

\noindent Noting that there are $1+\mathcal{O} \left(d_0^D\right)$ terms in the ball of radius $d_0$, the equation above yields

\begin{align}
    S_1\leq  \sum_{d(\mathbf{s},\mathbf{0}) \leq d_0} {1-r^{K_1 d(\mathbf{s},\mathbf{0})+K_2}} \leq (1-r) \sum_{d(\mathbf{s,0}) \leq d_0}(K_1 d(\mathbf{s},\mathbf{0})+K_2) \leq (1-r) (K_1d_0+K_2) \left(1+\mathcal{O}\left(d_0^D \right)\right),
\end{align}
\noindent which completes the proof.
\end{proof}

We will now bound the sum $S_2$ in the following Lemma:

\begin{lemma}\label{lemma:S2}
Consider the sum given by
\begin{align}
    S_2 := \max_{\mathbf{r'}} \sum_{d(\mathbf{r},\mathbf{r'}) > d_0} \left[1-r^{K_1d(\mathbf{r},\mathbf{r}')+K_2}\right]g(\mathbf{r},\mathbf{r'}),
\end{align}

\noindent where the function $g(\mathbf{r},\mathbf{r}')$ fulfills the condition in Eq.~\ref{eq:bound_g}. Then,

\begin{align}
S_2 \leq KC_D(d_0+1)^{D-1}\left(\zeta(\mu-D+1)-r^{K_1d_0+K_2}\mathrm{Li}_{\mu-D+1}\left(r^{K_1}\right)\right),
\end{align}
\noindent  where $\zeta(z)$ represents the Riemann zeta function and $\mathrm{Li}_s(z)$ the polylogarithm of order $s$. In particular, 

\begin{align}
S_2 \leq \mathcal{O} \left((1+d_0)^{D-1}(K_1d_0+K_2)\right) f(r),
\end{align}

\noindent with
\begin{align}
f(r) \leq 
\begin{cases} 
\mathcal{O}\left( (1-r)^{\mu-D}\right) & \text{if } D < \mu < D+1, \\ 
\mathcal{O}\left( (1-r) \log (1/(1-r)\right) & \text{if } \mu = D+1, \\ 
\mathcal{O}(1-r) & \text{if } \mu > D+1 .
\end{cases}
\end{align}
\end{lemma}
\begin{proof}

We write 
\begin{align}
    S_2= \max_{\mathbf{r'}} \sum_{d(\mathbf{r},\mathbf{r'}) > d_0} \left[1-r^{K_1d(\mathbf{r},\mathbf{r}')+K_2}\right]g(\mathbf{r},\mathbf{r'})\leq K \sum_{d(\mathbf{s},\mathbf{0}) > d_0} \frac{1-r^{K_1 d(\mathbf{r},\mathbf{r'})+K_2}}{\left[d(\mathbf{r},\mathbf{r}')-d_0\right]^{\mu}},
\end{align}

\noindent where we have used the bound on $g(\mathbf{r},\mathbf{r}')$ from Eq.~\ref{eq:bound_g}, and we have used translation invariance to sum over $\mathbf{s}$. Note that the $\mathbf{s}$ are summed on the lattice $\Lambda$,
\begin{align}\label{eq:S2_sum_d0}
S_2 \leq K \sum_{\substack{\mathbf{s} \in \Lambda  \\ d(\mathbf{s},\mathbf{0}) >d_0} } \frac{1-r^{K_1 d(\mathbf{s},\mathbf{0})+K_2}}{\left[d(\mathbf{s},\mathbf{0})-d_0\right]^{\mu}}.
\end{align}

\noindent In order to perform the sum, we note that, for each coordinate $s_i$, and for any function $h_i(\min(s_i,L-s_i))$, we can bound $\sum_{s_i=0}^{
L-1} h_i(\min(s_i,L-s_i)) \leq 2 \sum_{s_i=0}^{L/2} h_i(|s_i|)$. This allows us to replace the distance $d(\mathbf{s},\mathbf{0})$ with the $1-$norm $\|\mathbf{s}\|_1$, with the cost of an additional factor of $2^D$, and with the sum performed over  $\mathcal{C}=\{0, \dots, L/2\}^D$. Therefore, we take Eq.~\ref{eq:S2_sum_d0} and write
\begin{align}
S_2 \leq K2^D\sum_{\substack{{\mathbf{s}} \in \mathcal{C}  \\ \|\mathbf{s}\|_1 >d_0} } \frac{1-r^{K_1 \|\mathbf{s}\|_1+K_2}}{\left[\|\mathbf{s}\|_1-d_0\right]^{\mu}}\leq  K\sum_{\substack{ \mathbf{s} \in \mathds{N}^D  \\ \|\mathbf{s}\|_1 >d_0} } 2^D\frac{1-r^{K_1 \|\mathbf{s}\|_1+K_2}}{\left[\|\mathbf{s}\|_1-d_0\right]^{\mu}},
\end{align}

\noindent where we have extended the sum over all naturals, which will help provide an analytical bound. Since the sum only depends on $\|\mathbf{s}\|$, it can be performed as 
\begin{align}\label{eq:sum_ND(m)}
S_2 \leq 2^D \sum_{\substack{ \mathbf{s} \in \mathds{N}^D  \\ \|\mathbf{s}\|_1 >d_0} } \frac{1-r^{K_1 \|\mathbf{s}\|_1+K_2}}{\left[\|\mathbf{s}\|_1-d_0\right]^{\mu}} =\sum_{m=d_0+1}^{\infty} 2^DN_D(m)\frac{1-r^{K_1m+K_2}}{\left[m-d_0\right]^{\mu}},
\end{align}

\noindent where $N_D(m)$ represents the number of distinct vectors $\mathbf{s}\in \mathds{N}^D$ with $\|\mathbf{s}\|_1=m$. It can be bounded via standard combinatorial methods as 
\begin{align}\label{eq:N_D(m)}
2^DN_D(m) = 2^D\binom{m+D-1}{D-1}\leq 2^D \frac{(m+D-1)^{D-1}}{(D-1)!} \leq C_Dm^{D-1},
\end{align}

\noindent where $C_D=2^D(1+D)^{D-1}/(D-1)! = \mathcal{O}(1)$ is a dimension-dependent constant. We remark that this a crude bound that could be refined, but it is sufficient for our purpose.

We will now perform the change of variables $m'=m-d_0$ in Eq.~\ref{eq:sum_ND(m)}, and use the bound for $N_D(m)$ in Eq.~\ref{eq:N_D(m)} to write
\begin{align}
S_2 \leq K \sum_{m'=1}^{\infty} N_D(m'+d_0)\frac{1-r^{K_1(m'+d_0)+K_2}}{\left[m'\right]^{\mu}} \leq KC_D\sum_{m'=1}^{\infty} (m'+d_0)^{D-1}\frac{1-r^{K_1(m'+d_0)+K_2}}{\left[m'\right]^{\mu}}.
\end{align}

\noindent This sum converges when $\mu>D$. We will use $ (m'+d_0)^{D-1} \leq (m')^{D-1} (1+d_0)^{D-1}$ to compute

\begin{align}\label{eq:sum_ND(m_prime)}
S_2 &\leq K C_D \sum_{m'=1}^{\infty} \left(m'+d_0\right)^{D-1}\frac{1-r^{K_1(m'+d_0)+K_2}} {\left[m'\right]^{\mu}}\leq KC_D(d_0+1)^{D-1}\sum_{m'=1}^{\infty}\frac{1-r^{K_1(m'+d_0)+K_2}} {\left[m'\right]^{\mu}} \nonumber \\ 
&=KC_D(d_0+1)^{D-1}\left(\zeta(\mu-D+1)-r^{K_1d_0+K_2}\mathrm{Li}_{\mu-D+1}\left(r^{K_1}\right)\right),
\end{align}

\noindent where $\zeta(z)$ represents the Riemann zeta function, and $\mathrm{Li}_s(z)$ the polylogarithm of order $s$. Note that we can also write
\begin{align}
S_2 &\leq \mathcal{O} \left[(1+d_0)^{D-1}(K_1d_0+K_2)\right]f(r),
\end{align}

\noindent with
\begin{align}
f(r) =  
\begin{cases} 
\mathcal{O}\left((1-r)^{\mu-D} \right) & \text{if } D < \mu < D+1, \\ 
\mathcal{O}\left( (1-r) \log (1/(1-r)) \right) & \text{if } \mu = D+1, \\ 
\mathcal{O} \left(1-r \right)& \text{if } \mu > D+1 .
\end{cases}
\end{align}

\end{proof}

We now provide a Lemma that combines Lemmas \ref{lemma:S1} and \ref{lemma:S2} to provide a general bound for the sum $S$.

\begin{lemma}\label{lemma:bound_S}
Consider the sum given by
\begin{align}
    S := \max_{\mathbf{r'}} \sum_{ \mathbf{r}} \left[1-r^{K_1d(\mathbf{r},\mathbf{r}')+K_2}\right]g(\mathbf{r},\mathbf{r'}),
\end{align}

\noindent where the function $g(\mathbf{r},\mathbf{r}')$ fulfills the condition in Eq.~\ref{eq:bound_g}. Then,

\begin{align}
|S|\leq \mathcal{O} \left[(K_1d_0+K_2)(1+d_0^{D})\right]f(r),
\end{align}

\noindent with
\begin{align}
f(r) \leq 
\begin{cases} 
\mathcal{O}\left( (1-r)^{\mu-D}\right) & \text{if } D < \mu < D+1, \\ 
\mathcal{O}\left( (1-r) \log (1/(1-r)\right) & \text{if } \mu = D+1, \\ 
\mathcal{O}(1-r) & \text{if } \mu > D+1 .
\end{cases}
\end{align}
\end{lemma}

\begin{proof}
Follows directly from Lemmas \ref{lemma:S1} and \ref{lemma:S2}.
\end{proof}

Finally, we will provide two corollaries with specific choices of the parameters $K_1,K_2,d_0$, which will be directly used in the main text. First, we consider the case with $K_1=1,K_2=\mathcal{O}(1)$, and $d_0=0$.

\begin{corollary}\label{corollary:sum_1}
Consider the sum given by
\begin{align}
    S := \max_{\mathbf{r'}} \sum_{ \mathbf{r}} \left[1-r^{d(\mathbf{r},\mathbf{r}')+\varphi_0}\right]g(\mathbf{r},\mathbf{r'}),
\end{align}

\noindent where the function $g(\mathbf{r},\mathbf{r}')$ fulfills 
\begin{equation}
\left| g(\mathbf{r},\mathbf{r}') \right| \leq \frac{K}{d(\mathbf{r},\mathbf{r'})^{\mu}}, \quad\forall (\mathbf{r} \neq\mathbf{r'},\alpha,\beta).
\end{equation}

Then,
\begin{align}
|S| \leq (1-r)\varphi_0+KC_D\left(\zeta(\mu-D+1)-r^{\varphi_0}\mathrm{Li}_{\mu-D+1}\left(r\right)\right),
\end{align}

\noindent where $C_D=\mathcal{O}(1)$ is a dimension dependent constant, $\zeta(z)$ represents the Riemann zeta function, and $\mathrm{Li}_s(z)$ the polylogarithm of order $s$. Furthermore, assuming $K,\varphi_0=\mathcal{O}(1)$, 
\begin{align}
|S| \leq f(r),
\end{align}

\noindent with 
\begin{align}
f(r)= 
\begin{cases} 
\mathcal{O}\left( (1-r)^{\mu-D}\right) & \text{if } D < \mu < D+1, \\ 
\mathcal{O}\left( (1-r) \log (1/(1-r)\right) & \text{if } \mu = D+1, \\ 
\mathcal{O}(1-r) & \text{if } \mu > D+1 .
\end{cases}
\end{align}
\begin{proof}
Follows directly from setting $K_1=1$, $K_2=\varphi_0$, and $d_0=0$ in Lemma \ref{lemma:bound_S}.
\end{proof}

\end{corollary}

Finally, we provide a corollary for the case when $K_1=\mathcal{O}\left(d^{D-1}\right)$, $K_2=\mathcal{O}\left(d_0^D\right)$, and $d_0=\mathcal{O}(d)$, which is used in the main text to bound the stability of interacting systems in subsection \ref{subsection:noisy_evolution_interacting}.
\begin{corollary}\label{corollary:sum_2}
Consider the sum given by
\begin{align}
    S := \max_{\mathbf{r'}} \sum_{ \mathbf{r}} \left[1-r^{K_1d(\mathbf{r},\mathbf{r}')+K_2}\right]g(\mathbf{r},\mathbf{r'}),
\end{align}

\noindent where the function $g(\mathbf{r},\mathbf{r}')$ fulfills 
\begin{align}
|g(\mathbf{r},\mathbf{r'}) |\leq  
\begin{cases}
1  & \text{if } \; d(\mathbf{r},\mathbf{r'}) \leq d_0,\\
\frac{K}{\left[d(\mathbf{r},\mathbf{r'})-d_0\right]^{\mu}}  & \text{if } \; d(\mathbf{r},\mathbf{r'}) > d_0.
\end{cases}
\end{align}

Assume $K_1=\mathcal{O}\left(d^{D-1}\right), K_2=\mathcal{O}(d^D), d_0=\mathcal{O}(d)$, for some $d$. Then,
\begin{equation}
|S| \leq \mathcal{O}\left(d^{2D} \right)f(r),
\end{equation}

\noindent with 
\begin{align}
f(r) \leq 
\begin{cases} 
\mathcal{O}\left( (1-r)^{\mu-D}\right) & \text{if } D < \mu < D+1, \\ 
\mathcal{O}\left( (1-r) \log (1/(1-r)\right) & \text{if } \mu = D+1, \\ 
\mathcal{O}(1-r) & \text{if } \mu > D+1 .
\end{cases}
\end{align}
\begin{proof}
Follows directly from setting $K_1=\mathcal{O}\left(d^{D-1}\right), K_2=\mathcal{O}(d^D), d_0=\mathcal{O}(d)$ in Lemma \ref{lemma:bound_S}. 



    

\end{proof}

\end{corollary}

\section{Proof of bound for jump discontinuities}\label{appendix:jump_proof}
Here we will provide the proof for the bound on jump discontinuities that is used in the proof of Proposition \ref{proposition:result_discontinuities}. Precisely, we consider the error associated to a jump discontinuity, which following the notation in subsection \ref{subsection:1D_Fermi} is given by 
\begin{align}
E_j^{\mathrm{jump}}(k)=\frac{\Delta_j}{N^2} \sum_{x,y}\sum_q \left[1-(1-p)^{\varphi(x,y)}\right]e^{i(k-q)(x-y)}H(q-q_j),
\end{align}

\noindent where $\varphi(x,y)=1+d(x,y)$, and $d(x,y)=\min(|x-y|,N-|x-y|)$ represents the distance between $x$ and $y$. Since $H(q)$ is a Heaviside step function, we can directly write

\begin{align}
    \left|E_j^{\mathrm{jump}}(k)\right|=\left|\frac{\Delta_j}{N^2} \sum_{x,y}\sum_{q\geq q_j} \left[1-(1-p)^{\varphi(x,y)}\right]e^{i(k-q)(x-y)}\right| \leq \left|\frac{2}{N}\sum_{s} \left[1-(1-p)^{1+d(s,0)}\right]G_j(s,k)\right|,
\end{align}

\noindent where we have performed the change of variables $s=x-y$, applied $|\Delta_j| \leq 1$, and defined
\begin{align}
G_j(s,k):=\sum_{q \geq q_j}e^{i(k-q)s}.
\end{align}

\noindent We will now provide a technical lemma that allow us to bound the behavior of $E_j^{\mathrm{jump}}(k)$.

\begin{lemma}\label{lemma:Fourier_coefficients}
Consider the function $\phi(\theta)=\min(\theta/2\pi,1-\theta/2\pi)$ with $\theta \in [0,2\pi]$. Define the function
\begin{equation}
f(\theta):=\left(\frac{1}{N}+\phi(\theta)\right)r^{N\phi(\theta)},
\end{equation}

\noindent with $r \in [0,1]$, and its Fourier coefficients
\begin{equation}
F_m :=\frac{1}{2\pi}\int_0^{2\pi}f(\theta)e^{-i\theta m} \, d\theta.
\end{equation}

\noindent Then,
\begin{equation}
|F_m| \leq \frac{2}{\pi^2 m^2}\left(1+|\log(r)|\right)=\mathcal{O}\left(\frac{1}{m^2}\right).
\end{equation}

\end{lemma}
\begin{proof}
By definition $f(\theta)$ is periodic in $\theta$ with period $2\pi$. We can then integrate by parts

\begin{align}\label{eq:F_m_parts}
F_m=\frac{1}{2\pi}\int_0^{2\pi}f(\theta)e^{-i\theta m} \, d\theta = \frac{1}{i2 \pi m}\int_0^{2\pi}f'(\theta)e^{-i\theta m} \, d\theta=\frac{1}{2\pi m^2}\left(f'(2\pi)-f'(0)\right)-\frac{1}{2 \pi m^2}\int_0^{2\pi}f''(\theta)e^{-i \theta m} \, d \theta,
\end{align}

\noindent where the first boundary term vanishes due to the periodicity of $f(\theta)$. We compute the first and second derivatives as

\begin{align*}
    f'(\theta)&=\phi'(\theta) r^{N \phi(\theta)} \left(1+ \log r +N \phi(\theta) \log r\right),\\
    f''(\theta)&=\frac{N \log r}{4 \pi^2} r^{N \phi(\theta)}\left(2+ \log r+N \phi(\theta) \log r\right).
\end{align*}

\noindent Note that $f'(2\pi)-f'(0) =-(1+\log r)/\pi$. Furthermore, by straightforwards integration, one finds that

\begin{equation}
\int_0^{2\pi} |f''(\theta)| \, d\theta \leq \frac{N |\log r|}{2 \pi^2} \left(|2+ \log r|\int_0^{\pi} r^{N \phi(\theta) } d \theta+N\phi |\log r| \int_0^{\pi} r^{N \phi(\theta)} \phi(\theta) d\theta\right) \leq \frac{3+|\log r|}{ \pi}.
\end{equation}
\noindent Then, using Eq.~\ref{eq:F_m_parts} we directly bound

\begin{equation}
|F_m| \leq \frac{2}{\pi^2 m^2}\left(1+|\log r|\right).
\end{equation}
\end{proof}

Using this result, we are able to upper-bound $|E_j^{\mathrm{jump}}(k)|$. For simplicity, we will quantize the momenta as $q_m=2\pi m/N$ (i.e., we assume periodic boundary conditions).

\begin{lemma}\label{lemma:sup_norm_Sj}
Consider the sum 
\begin{equation}
E_j^{\mathrm{jump}}(k,p)=\frac{1}{N}\sum_{s=0}^{N-1} \left(1-(1-p)^{1+d(s,0)}\right)\sum_{m=j}^{N/2-1} e^{is(q_k-q_m)},
\end{equation}

\noindent where $q_m=2\pi m/N$. Then,
\begin{equation}
|E_{j}^{\mathrm{jump}}(k,p)| = \mathcal{O}\left(\frac{p}{|k-q_j|}\right)
\end{equation}

\noindent when $q_j \neq k$, and $|E_{j}^{\mathrm{jump}}(k,p)|=\mathcal{O}(pN)$ when $q_j=k$.
\end{lemma}

\begin{proof}
We will consider the case with $k=0$. Note that the general case automatically follows due to periodicity, since $k$ is merely shifting the phase. Let us denote by $B_j$ the derivative with respect to $p$ of $E_j^{\mathrm{jump}}(k=0,p)$,
\begin{align}
B_j(p):=\frac{1}{N}\sum_{s=0}^{N-1}\left(1+d(s,0)\right)(1-p)^{d(s,0)}\sum_{m=j}^{N/2-1}e^{-isq_m}.
\end{align}

\noindent We will now provide a bound on $|B_j(p)|$, which in turn will allow us to bound $|S_j(k=0,p)|$.
Let us now denote $\varphi(s)=N\phi(\theta_s)$, where $\theta_s=2\pi s/N$. We will then write
\begin{equation}
B_j(p)=\sum_{s=0}^{N-1} \sum_{m=j}^{N/2-1} f(\theta_s) e^{-isq_m},
\end{equation}

\noindent where $f(\theta_s)=(1/N+\phi(\theta_s))(1-p)^{N\phi(\theta_s)}$ and $\phi(\theta_s)=\min(\theta_s/2\pi,1-\theta_s/2\pi)$.

 Let us denote by $F_\nu$ the Fourier coefficients of $f(\theta_s)$. Naturally, one can always write $f(\theta_s)=\sum_{\nu \in\mathds{Z}} F_\nu e^{i\nu \theta_s}$. Taking into account that $q_m=2\pi m/N$, we can then sum
\begin{equation}
\sum_{s=0}^{N-1} f(\theta_s)e^{-i \theta_s m}=\sum_{s=0}^{N-1} \sum_{\nu \in \mathds{Z}}F_\nu e^{i \theta_s (\nu-m)}=\sum_{\nu \in \mathds{Z}}F_{\nu}N\delta_{(\nu-m) (\mathrm{mod} N),0}=N\sum_{\nu \in \mathds{Z}}F_{m+\nu N}.
\end{equation}

\noindent Hence, we obtain that 
\begin{align}\label{eq:expression_Bj}
B_j(p)=N\sum_{m=j}^{N/2-1} \sum_{h \in \mathds{Z}}F_{m+hN}.
\end{align}

\noindent Note that, from Lemma \ref{lemma:Fourier_coefficients}, $|F_\nu| \leq \mathcal{O}(1/\nu^2)$. When $k> 0$ we can then bound
\begin{align}\label{eq:bound_Fm}
\sum_{m=k}^{\infty}|F_m| \leq \sum_{m=k}^{\infty}\frac{2}{\pi^2 m^2}\left(1+|\log (1-p)|\right) \leq \frac{2}{\pi^2}\left(1+|\log(1-p)|\right)\frac{1}{k-1}=\mathcal{O}\left(\frac{1}{k}\right), 
\end{align}

\noindent where we have assumed $|\log(1-p)|=\mathcal{O}(1)$, which can be satisfied by upper-bounding the error rate as $p< 1$. Consider now the case with $j>0$. Then, following the argument in Eq.~\ref{eq:bound_Fm} we directly bound
\begin{align}
\left|\sum_{m=j}^{N/2-1}F_m \right| \leq \sum_{m=j}^{\infty}|F_m| \leq \mathcal{O}\left(\frac{1}{j}\right).
\end{align}

One the other hand, when $j<0$ we can bound

\begin{align}
\left|\sum_{m=j}^{N/2-1}F_m \right|= \left| \sum_{m=-\infty}^{\infty}F_m-\sum_{m=-\infty}^{j}F_m - \sum_{m=N/2-1}^{\infty}F_m \right|\leq \mathcal{O}(1/N)+\mathcal{O}(1/j)+\mathcal{O}(1/N),
\end{align}

\noindent where we have used that $\sum_{m=-\infty}^{\infty} F_m=1/N$ for the left-most term, and the bound from Eq.~\ref{eq:bound_Fm} for the other two. Hence, we conclude that
\begin{align}\label{eq:bound_Oj}
    \left|\sum_{m=j}^{N/2-1}F_m \right| \leq \mathcal{O}\left(\frac{1}{j}\right), \forall j\neq 0.
\end{align}

Plugging the bound in Eq.~\ref{eq:bound_Oj} into Eq.~\ref{eq:expression_Bj} we can control the contribution of all the terms with $h \neq 0$, thus obtaining
\begin{align}\label{eq:Bound_Bj}
|B_j(p)| \leq  N\underbrace{\left|\sum_{m=j}^{N/2-1}\sum_{h \neq 0}F_{m+hN}\right|}_{\mathcal{O}(1/N)}+  N \underbrace{\left|\sum_{m=j}^{N/2-1} F_m \right|}_{\mathcal{O}(1/j)}=\mathcal{O}(1)+\mathcal{O}(N/j),
\end{align}

\noindent where the contribution for the terms with $h \neq 0$ can directly be bounded to be $\mathcal{O}(1)$ via Eq.~\ref{eq:bound_Fm}. Together with Eqs.~\ref{eq:Bound_Bj} and \ref{eq:bound_Fm} this allows us to bound 
\begin{align}
|B_j(p)|\leq \mathcal{O}(N/j).
\end{align}

Hence, we conclude that $|B_j(p)|=\mathcal{O}(N/j),\forall p \in [0,1)$. We can then write
\begin{align}
 E_j^{\mathrm{jump}}(k=0,p)-E_j^{\mathrm{jump}}(k=0,p=0)=\int_0^p \frac{\partial}{\partial p} E_j^{\mathrm{jump}}(k=0,p) \, dp= \int_0^p B_j(p) \, dp.
 \end{align}

 \noindent Hence, we bound
\begin{align}
 \left|E_j^{\mathrm{jump}}(k=0,p)\right|\leq \int_0^p |B_j(p)| \, dp \leq \mathcal{O}(pN/j),
 \end{align}

which completes the proof.

\end{proof}

\section{Proof of bound for Lipschitz continuous functions}\label{appendix:lipschitz}

We will here bound
\begin{align}
E_L(k)=\frac{1}{N^2}\sum_{x,y}\sum_q\left[1-(1-p)^{\varphi(x,y)}\right]e^{i(k-q)(x-y)}n_L(q),
\end{align}

\noindent where the momentum occupation $n_L(q)$ is an $L-$Lipchitz continuous function. Let us perform the change of variables $s=x-y$ to write
\begin{align}\label{eq:Sum_EL}
E_L(k)=\frac{2}{N}\sum_{s=1}^{N-1}\sum_q\left[1-(1-p)^{\varphi(s,0)}\right]e^{i(k-q)s}n_L(q)+\frac{p}{N}\sum_q n_L(q),
\end{align}

\noindent where the momenta are summed over the Brillouin zone. That is, we will consider the momentum $q_j$, with $j\in \{-N/2,\dots N/2-1\}$ a label for the $N$ modes. As stated in Eq.~\ref{eq:momentum_modes}, we assume that $q_{j+1}=q_j+2\pi/N$. Furthermore, we assume periodic boundary conditions on $n_L(q)$, implying $n_L(q_{-N/2})=n_L(q_{N/2})$, where $q_{N/2}=q_{-N/2}+2\pi$.

Now, we will provide a technical lemma that bounds the sum.

\begin{lemma}\label{lemma:1norm_Sj}
Consider the sum 
\begin{align}\label{eq:sum_lemma_EL(k)}
E_L(k)=\frac{2}{N}\sum_{s=1}^{N-1}\sum_q\left[1-(1-p)^{\varphi(s,0)}\right]e^{i(k-q)s}n_L(q)+\frac{p}{N}\sum_q n_L(q),
\end{align}

\noindent and assume that $n_L(q)$ is an $L-$Lipschitz continuous function. Then, 
\begin{equation}
|E_L(k)| \leq {\mathcal{O}} \left( L\sqrt{p} \right)
\end{equation}
\end{lemma}

\begin{proof}

Consider the sum 
\begin{align}\label{eq:A(k)_def}
A(k)=\frac{1}{N}\sum_{s=1}^{N-1}\sum_q\left[1-(1-p)^{\varphi(s,0)}\right]e^{i(k-q)s}n_L(q).
\end{align}

Denote the Fourier transform of $n_L(q)$ as
\begin{align}
\tilde{n}_L(s):=\frac{1}{N}\sum_q e^{-iqs}n_L(q),
\end{align}

\noindent and denote also $g(s):=\left[1-r^{\varphi(s,0)}\right]$, with $r=1-p$. Using Cauchy-Schwarz inequality, we can bound Eq.~\ref{eq:A(k)_def} as
\begin{align}\label{eq:bound_A(k)}
|A(k)|=\left|\sum_{s=1}^{N-1} g(s)e^{iks}\tilde{n}_L(s)\right|  \leq \sqrt{\sum_{s=1}^{N-1}\frac{g(s)^2}{\sin^2\left(\pi s/N\right)}}\sqrt{\sum_{s=1}^{N-1}\sin^2 \left(\pi s/N \right) |\tilde{n}_L(s)|^2}.
\end{align}

\noindent We will bound both sums separately. First, consider the Fourier transform of the finite differences function $n_L(q_{j+1})-n_L(q_j)$, which we denote as $\tilde{F}(s)$. We can write
\begin{align}
\tilde{F}(s)=\frac{1}{N} \sum_{j=-N/2}^{N/2-1} e^{-isq_j} \left[n_L(q_{j+1})-n_L(q_j)\right]=\left(e^{i\frac{2\pi s}{N}}-1\right)\tilde{n}(s).
\end{align}

\noindent This Fourier transform is a well-known relation that follows from the fact that $q_{j+1}=q_j+2\pi/N$, and can be shown by assuming the periodic boundary condition $n(q_{-N/2})=n(q_{N/2})$. This allows us to write $|\tilde{F}(s)|$ as
\begin{align}\label{eq:relation_Fs_sin}
\left|\tilde{F}(s)\right|^2=\left|e^{i\frac{2\pi s}{N}}-1\right|^2 \left| \tilde{n}(s)\right|^2=4 \sin ^2 \left(\frac{\pi s}{N}\right)\left|\tilde{n}(s)\right|^2.
\end{align}

\noindent Furthermore, due to the Lipschitz continuity assumption for $n_L(q)$ we can bound
\begin{align}
\left|n_L(q_{j+1})-n_L(q_j)\right| \leq L |q_{j+1}-q_j| = \frac{2\pi L}{N}.
\end{align}

\noindent As a consequence, we can use Parseval's theorem to bound
\begin{align}\label{eq:parseval_bound}
\sum_{s=1}^{N-1}\left|\tilde{F}(s)\right|^2=\frac{1}{N} \sum_{j=-N/2}^{N/2-1} \left|n_L(q_{j+1})-n_L(q_j)\right|^2 \leq \frac{1}{N} \sum_{j=-N/2}^{N/2-1} \frac{4\pi^2 L^2}{N^2}=\frac{4\pi^2 L^2}{N^2}.
\end{align}

Therefore, using Eqs.~\ref{eq:relation_Fs_sin} and ~\ref{eq:parseval_bound} we can write
\begin{align}\label{eq:left_bound}
\sum_{s=1}^{N-1} \sin^2 \left(\pi s/N \right) |\tilde{n}_L(s)|^2 =\frac{1}{4} \sum_{s=1}^{N-1}\left|\tilde{F}(s)\right|^2\leq \frac{\pi^2 L^2}{N^2}.
\end{align} 

\noindent Let us now focus on the left hand sum in the bound of Eq.~\ref{eq:bound_A(k)}. Use the definition of $g(s)$ to write
\begin{align}\label{eq:g(s)_square}
\sum_{s=1}^{N-1}\frac{g(s)^2}{\sin^2\left(\pi s/N\right)}=\sum_{s=1}^{N-1}\frac{\left(1-r^{\varphi(s,0)}\right)^2}{\sin^2\left(\pi s/N\right)}\leq2\sum_{s=1}^{N/2}\frac{\left(1-r^{\varphi_0+s}\right)^2}{\sin^2\left(\pi s/N\right)}\leq \frac{N^2}{2}\sum_{s=1}^{N/2}\frac{(1-r^{\varphi_0+s})^2}{s^2},
\end{align}

\noindent where we have first used $\varphi(s,0)=\varphi_0+d(s,0)$, with the distance $d(s,0)=s$ when $0 \leq s\leq N/2$, and we have also used Jordan's inequality $\sin(\pi s /N)\geq 2s/N$ when $0 \leq s \leq N/2$. Note that the sum in Eq.~\ref{eq:g(s)_square} can be bounded as a dilogarithm, analogous to Lemma \ref{lemma:S2} in Appendix \ref{appendix:sum}. We can write 
\begin{align}\label{eq:right_bound}
\sum_{s=1}^{N-1}\frac{g(s)^2}{\sin^2\left(\pi s/N\right)} \leq \frac{N^2}{2}\sum_{s=1}^{\infty}\frac{(1-r^{\varphi_0+s})^2}{s^2} = \frac{N^2}{2} \left(\zeta(2)-2r^{\varphi_0}\mathrm{Li}_2(r)+r^{2\varphi_0}\mathrm{Li}\left(r^{2}\right)\right) = \frac{N^2}{2} \mathcal{O}\left(1-r\right),
\end{align}

\noindent where we have expanded $\mathrm{Li}_2(r)=\left[1-\log(1-r)\right](1-r)+\mathcal{O} ((1-r)^2)$, and used $\varphi_0=\mathcal{O}(1)$. In the expression above, $\zeta(z)$ refers to Riemann's zeta function, and $\mathrm{Li}_2(z)$ to the dilogarithm.

Now, plugging the bounds in Eqs.~\ref{eq:right_bound} and \ref{eq:left_bound} into Eq.~\ref{eq:bound_A(k)} yields
\begin{align}
|A(k)| \leq \sqrt{\sum_{s=1}^{N-1}\frac{g(s)^2}{\sin^2\left(\pi s/N\right)}}\sqrt{\sum_{s=1}^{N-1}\sin^2 \left(\pi s/N \right) |\tilde{n}_L(s)|^2} \leq \sqrt{N^2\mathcal{O}(1-r)} \frac{2\pi L}{N}=\mathcal{O}\left(L \sqrt{1-r}\right).
\end{align}

Finally, since from Eq.~\ref{eq:sum_lemma_EL(k)} $E_L(k)=A(k)+p N_{\mathrm{occ}}/N$, we obtain $|E_L(k)| \leq  \mathcal{O}(L\sqrt{p})$, completing the proof.
\end{proof}

\section{2D Fermi surface in the thermodynamic limit}\label{appendix:2D_Fermi}

In this section we will analytically study the stability to errors in 2D fermionic states with Fermi surfaces that was numerically demonstrated in subsection \ref{subsection:2D_Fermi}. For simplicity, we will consider the thermodynamic limit (in which the system size tends to infinity $N \to \infty$) of states with a circular Fermi surface. Precisely, we will consider a state whose momentum occupation satisfies
\begin{align}
	n(\mathbf{q}) = \begin{cases}
		                1 & \text{if } \|\mathbf{q}\|_2 \leq k_F \ \\
		                0 & \text{if } \|\mathbf{q}\|_2>k_F.
	                \end{cases}
\end{align}

\noindent This Fermi surface occurs in tight-binding models, such as the one presented in subsection \ref{subsection:2D_Fermi}, in the low filling limit. Hence, we will assume $k_F \ll \pi$.

We aim to bound the error in measuring the occupation $n_{\mathbf{k}} = a_{\mathbf{k}}^\dagger a_{\mathbf{k}}$ for a given momentum $\mathbf{k}$. We will here assume that the encoding maps the Majorana bilinear $\gamma_\mathbf{r}^{\alpha} \gamma_\mathbf{r'}^{\alpha'}$ to a Pauli string with support $\|\mathbf{r}-\mathbf{r'}\|_2$. Note that, unlike in the previous sections, here the distance is given by the $2-$norm instead of the $1-$norm. This is done to simplify the presentation and calculation, and will not significantly alter the final results. Throughout the section, and to simplify the notation, we will often denote the $2$-norm as $\|\mathbf{s}\|_2=|\mathbf{s}|=s$. Furthermore, we have not used here the periodic boundary conditions for the distance, since we are dealing with the thermodynamic limit. Under these assumptions, and taking the infinite $N$ limit in Eqs.~\ref{error_expression_DD}, \ref{eq:Np_Ok}, the error is given by:
\begin{align}\label{eq:error_integral_2D_limit}
	\mathrm{Error}(\mathbf{k}) = |\langle n_{\mathbf{k}}\rangle_{\mathrm{noisy}} - \langle n_{\mathbf{k}}\rangle_{\mathrm{noiseless}}| = \left|\int d^2\mathbf{s} \, \left(1-e^{-\lambda |\mathbf{s}|}\right) e^{i \mathbf{k} \cdot \mathbf{s}} C(\mathbf{s})\right|
\end{align}
where $C(\mathbf{s}) = \langle c_{\mathbf{r}}^\dagger c_{\mathbf{r}+\mathbf{s}} \rangle$ is the two-point correlation function, and $\lambda=-\log(1-p)$.

Let us write the noisy expectation value as 
\begin{align}
\langle n_{\mathbf{k}}\rangle_{\mathrm{noisy}}=\int_{\mathds{R}^2} d^2 \mathbf{s} e^{-\lambda |\mathbf{s}|}e^{i \mathbf{k} \cdot \mathbf{s}}C(\mathbf{s})=\frac{1}{(2\pi)^2}\int_{|\mathbf{q}| \leq k_F} d^2 \mathbf{q} \int_{\mathds{R}^2} d^2 \mathbf{s} e^{-\lambda |\mathbf{s}|}e^{i (\mathbf{k-\mathbf{q})} \cdot \mathbf{s}}.
\end{align}

\noindent Let us first perform the integral in $\mathbf{s}$ by going to polar coordinates, denoting $\mathbf{P}=\mathbf{k}-\mathbf{q}$ and computing 
\begin{align}
\int_{\mathds{R}^2} d^2 \mathbf{s} e^{-\lambda |\mathbf{s}|}e^{i (\mathbf{k-\mathbf{q})} \cdot \mathbf{s}}=\int_0^{2\pi} d \theta \int_0^{\infty }sds e^{-\lambda s+iPs \cos (\theta)} =\int_0^{2\pi}d \theta \frac{1}{\left[\lambda-iP \cos \theta\right]^2} = \frac{2 \pi \lambda}{\left[\lambda^2+P^2\right]^{3/2}},
\end{align}

\noindent when $\lambda^2+P^2 \neq 0$. Hence, we obtain
\begin{align}
\label{eq:n_noisy_tmp}
\langle n_{\mathbf{k}}\rangle_{\mathrm{noisy}}=\frac{1}{2\pi} \int_{|\mathbf{q}| \leq k_F} d^2 \mathbf{q} \frac{\lambda}{\left[\lambda^2+|\mathbf{k}-\mathbf{q}|^2\right]^{3/2}}.
\end{align}

\noindent Now, consider the case where $\mathbf{k}$ lies outside the Fermi surface (i.e., $|\mathbf{k}| >k_F$). Note that in this case $|\mathbf{k}-\mathbf{q}| \geq \Delta=|\mathbf{k}|-k_F$.

Consider again the variable $\mathbf{P}=\mathbf{k}-\mathbf{q}$, and note that in this case $|\mathbf{P}| > \Delta=|\mathbf{k}|-k_F $. Then, we can integrate in polar coordinates as
\begin{align}\label{eq:case_outside}
\langle n_{\mathbf{k}}\rangle_{\mathrm{noisy}}\leq\frac{1}{2\pi} \int_0^{2\pi} d\theta \int_{0}^{k_F} d q \frac{\lambda q}{\left[\lambda^2+\Delta^2\right]^{3/2}}= \int_{0}^{k_F} \frac{ \lambda q d q }{\left[\lambda^2+\Delta^2\right]^{3/2}}=\frac{1}{2}\frac{\lambda k_F^2}{\left[\lambda^2+\Delta^2\right]^{3/2}}  .
\end{align}

\noindent Note that, when $\mathbf{k}$ lies outside the Fermi surface, the noiseless expectation value is directly $\langle n_{\mathbf{k}}\rangle_{\mathrm{noiseless}}=0$. Then, we bound the error as
\begin{align}
\mathrm{Error}(\mathbf{k})= \frac{1}{2}\frac{\lambda k_F^2}{\left[\lambda^2+\Delta^2\right]^{3/2}} \leq \frac{\lambda k_F^2}{2 \Delta^3}  =\frac{\mathcal{O}(p)k_F^2}{\left(|\mathbf{k}|-k_F\right)^3}.
\end{align}

Now, consider the case where $\mathbf{k}$ lies inside the Fermi surface (i.e. $|\mathbf{k}| <k_F$). We will exploit the fact that the integral in Eq.~\ref{eq:n_noisy_tmp} over the entire $\mathds{R}^2$ space is $1$,
\begin{align}
\frac{1}{2\pi} \int_{\mathds{R}^2} d^2 \mathbf{q} \frac{\lambda}{\left[\lambda^2+|\mathbf{k}-\mathbf{q}|^2\right]^{3/2}}=1.
\end{align}

\noindent Here, we have implicitly assumed that the Brillouin zone is infinite, which is a good approximation for small occupation numbers which yield a circular Fermi surface. Furthermore, in this case $\langle n_{\mathbf{k}}\rangle_{\mathrm{noiseless}}=1$. We can then write
\begin{align}
 \left| \langle n_{\mathbf{k}}\rangle_{\mathrm{noiseless}} - \langle n_{\mathbf{k}}\rangle_{\mathrm{noisy}}\right|=1-\frac{1}{2\pi} \int_{|\mathbf{q}| \leq k_F} d^2 \mathbf{q} \frac{\lambda}{\left[\lambda^2+|\mathbf{k}-\mathbf{q}|^2\right]^{3/2}}=\frac{1}{2\pi}\int_{|\mathbf{q}| > k_F} d^2 \mathbf{q} \frac{\lambda}{\left[\lambda^2+|\mathbf{k}-\mathbf{q}|^2\right]^{3/2}}.
\end{align}

\noindent Then, the proof is identical to the case with $|\mathbf{k}| > k_F$, and one can use the technique in Eq.~\ref{eq:case_outside} to obtain
\begin{align}
\mathrm{Error}(\mathbf{k})= \frac{\mathcal{O}(p)k_F^2}{\left(k_F-|\mathbf{k}|\right)^3}.
\end{align}

Hence, the measuring task is stable to noise when measuring away from the Fermi surface. Now let us consider the case when $|\mathbf{k}|=k_F$, where the measurement is done exactly in the Fermi surface.  In this case, the boundary condition for the integral in Eq.~\ref{eq:n_noisy_tmp} becomes $|\mathbf{k}-\mathbf{P}| \leq k_F$ as a function of $\mathbf{P}=\mathbf{k}-\mathbf{q}$. Parameterizing $\mathbf{P}=P(\cos \theta,\sin \theta)$, the integral becomes 
\begin{align}\label{eq:noisy_k=kF}
\langle n_{\mathbf{k}}\rangle_{\mathrm{noisy}}=\frac{1}{2\pi} \int_{-\pi/2}^{\pi/2} d\theta \int_{0}^{2k_F \cos \theta} d P \frac{\lambda P}{\left[\lambda^2+P^2\right]^{3/2}} =\frac{1}{2\pi} \int_{-\pi/2}^{\pi/2} d\theta \left( \frac{ \lambda }{\sqrt{\lambda^2+4k_F^2 \cos^2 \theta}}-1 \right) =\frac{1}{2}+I(\lambda),
\end{align}

\noindent where
\begin{align}
I(\lambda)=\frac{\lambda}{\pi} \int_{-\pi/2}^{0}  \frac{ d\theta }{\sqrt{\lambda^2+4k_F^2 \cos^2 \theta}}=\frac{\lambda}{\pi}\int_{0}^{\pi/2}  \frac{ d\theta }{\sqrt{\lambda^2+4k_F^2 \sin^2 \theta}},
\end{align}

\noindent where we have shifted the integration variable by $\pi/2$ and used $\sin^2 \theta=\cos^2(\theta+\pi/2)$. Now, we will use Jordan's inequality, which states that $\sin \theta \geq 2 \theta /\pi$ when $\theta \in [0,\pi/2]$, to write
\begin{align}
I(\lambda)=\frac{\lambda}{\pi}\int_{0}^{\pi/2}  \frac{ d\theta }{\sqrt{\lambda^2+4k_F^2 \sin^2 \theta}} \leq \frac{\lambda}{\pi}\int_{0}^{\pi/2}  \frac{ d\theta }{\sqrt{\lambda^2+16(k_F/\pi)^2 \theta^2 }}.
\end{align}

\noindent Now we make the change of variables $u=(4k_F \theta)/(\pi\lambda)$ to write
\begin{align}
I(\lambda) \leq \frac{\lambda}{4k_F}\int_{0}^{2k_F/\lambda}  \frac{ d u }{\sqrt{1+u^2 }}= \frac{\lambda}{4 k_F} \mathrm{arcsinh}\left(\frac{2k_F}{\lambda}\right).
\end{align}

\noindent Hence, we obtain that $\lim_{\lambda \to 0} I(\lambda)=0$, and therefore Eq.~\ref{eq:noisy_k=kF} allows us to bound the error as 
\begin{align}
\lim_{p \to 0} \mathrm{Error}(\mathbf{k}) = \lim_{p \to 0} \left| \langle n_{\mathbf{k}}\rangle_{\mathrm{noiseless}} - \langle n_{\mathbf{k}}\rangle_{\mathrm{noisy}}\right| = \frac{1}{2} - \lim_{\lambda \to 0 }I(\lambda)=\frac{1}{2}.
\end{align}

\noindent We therefore conclude that, when $|\mathbf{k}|=k_F$, one can obtain a constant error even with vanishing error rates, and therefore stability to noise is not present. We conclude that, as demonstrated numerically in section \ref{subsection:2D_Fermi} and consistent with the $1D$ case, one obtains fragility to noise when measuring on the Fermi surface and stability elsewhere, with the error governed by 
\begin{align}
\mathrm{Error}(\mathbf{k})= \mathcal{O} \left(\frac{p k_F^2}{|k_F-|\mathbf{k}||^{3/2}}\right).
\end{align}

\section{Proof of correlation decay}\label{appendix:correlations}

In this appendix we will provide the proof of Lemma \ref{lemma:correlations_interacting} of the main text. That is, we will consider the noisy evolved state $\rho_{(k)}^{\mathrm{noisy}}=\Phi_k ^{\mathrm{id}} \Phi_{k-1}^{\mathrm{noisy}} \dots \Phi_{1}^{\mathrm{noisy}}(\rho_0)$, and want to bound correlations of the form 
\begin{align}
    \left|\langle  O_AO_B\rangle_{\rho_{(k)}^{\mathrm{noisy}}}\right|,
\end{align}

\noindent for $O_A$ and $O_B$ supported in disjointed regions $A$ and $B$, respectively. Note that we assume that $O_A$ and $O_B$ are odd fermionic operators, so that $\langle  O_A\rangle_{\rho_{(k)}^{\mathrm{noisy}}}=\langle  O_B\rangle_{\rho_{(k)}^{\mathrm{noisy}}}=0$ due to parity superselection rules.
We assume that the correlations in the initial state decay as 
\begin{align}
    \left|\langle  O_AO_B\rangle_{\rho_{(k)}^{\mathrm{noisy}}}\right| \leq f(d(A,B))
\end{align}

\noindent for some function $f(d(A,B))$ of the distance between $A$ and $B$. Now, let us consider the time-evolved operator
\begin{align}
    O_A(k)=\left[\Phi_1^{\mathrm{noisy}}\right]^{\dagger} \dots \left[\Phi_{k-1}^{\mathrm{noisy}}\right]^{\dagger}\left[\Phi_k^{\mathrm{id}}\right]^{\dagger}O_A.
\end{align}

\noindent Due to the locality of each layer of unitaries and standard Lieb-Robinson bounds, it follows that $O_A(k)$ is entirely contained in a region $\mathcal{B}_A(k)$ of size $|\mathcal{B}_A(k)|\leq |A|+\mathcal{O}(k^D)$ determined by the lightcone. Likewise, $O_B(k)$ is entirely contained in a region $\mathcal{B}_B(k)$ centered around subregion $B$, with radius $\mathcal{O}(k)$ and size $|\mathcal{B}_B(k)|\leq |B|+\mathcal{O}(k^D)$. That is, locality imposes a lightcone on which both $O_A(k)$ and $O_B(k)$ are supported.
We can write Heisenberg-evolved operator $[O_AO_B](k)$ as
\begin{align}\label{eq:expansion_OAOB(k)}
    \left[O_AO_B\right](k)=\left[\Phi_1^{\mathrm{noisy}}\right]^{\dagger} \dots \left[\Phi_{k-1}^{\mathrm{noisy}}\right]^{\dagger}\left[\Phi_k^{\mathrm{id}}\right]^{\dagger}(O_AO_B).
\end{align}

The proof will rely on the fact that the operator at time $k$ can be written as a convex combination of operators supported in $\mathcal{B}_A(k) \cup \mathcal{B}_B(k)$. That is,
\begin{align}\label{eq:expansion_OAOB(k)_local}
    \left[O_AO_B\right](k)=\left[\Phi_1^{\mathrm{noisy}}\right]^{\dagger} \dots \left[\Phi_{k-1}^{\mathrm{noisy}}\right]^{\dagger}\left[\Phi_k^{\mathrm{id}}\right]^{\dagger}(O_AO_B)=\sum_{i}p_i O_i^{\mathcal{B}_A(k)} O_i^{\mathcal{B}_B(k)},
\end{align}

\noindent where $p_i \geq 0$, $\sum_i p_i=1$, and $\supp \left(O_i^{\mathcal{B}_A(k)}\right) \in \mathcal{B}_A(k)$ and $\supp \left(O_i^{\mathcal{B}_B(k)}\right) \in \mathcal{B}_B(k)$. Furthermore $\|O_i^{\mathcal{B}_A(k)}O_i^{\mathcal{B}_B(k)}\| \leq \|O_A O_B\|$.

This can be shown by induction. Assume that, at a given time-step $\tau$, the expression in Eq.~\ref{eq:expansion_OAOB(k)_local} holds. Then, we will study the action of the channel $\left[\Phi_{\tau+1}^{\mathrm{noisy}}\right]^{\dagger}=\mathcal{U}_{\tau+1}^{\dagger} \circ  \mathcal{N}_p$. First, consider the action of the noise channel $\mathcal{N}_p (O_{\mathcal{B}_A(\tau)} O_{\mathcal{B}_B(\tau)})$, where $\supp \left(O_{\mathcal{B}_A(\tau)}\right)\in \mathcal{B}_A(\tau),\,\supp \left(O_{\mathcal{B}_B(\tau)}\right)\in \mathcal{B}_B(\tau)$. We will show that the channel $\mathcal{N}_p$ will not change the support of the operator. To see this, we can write the noise channel as 
\begin{align}
\mathcal{N}_p(\cdot)=\sum_j p_j S_j (\cdot)S_j,
\end{align}

\noindent where $S_j$ are Pauli strings. This is possible since $\mathcal{N}_p$ is a Pauli noise channel. Note that, by Assumption \ref{assumption:encodings}, a Pauli string $S_j$ is mapped to a Majorana string $\Gamma_j$ (i.e., a product of Majoranas). Then, the action of the noise channel can be written as 
\begin{align}\label{eq:noise_channel_OAOB}
    \mathcal{N}_p \left(O_{\mathcal{B}_A(\tau)} O_{\mathcal{B}_B(\tau)} \right)=\sum_{j}p_j \left( \Gamma_jO_{\mathcal{B}_A(\tau)} O_{\mathcal{B}_B(\tau) } \Gamma_j\right)=\sum_jp_j \widetilde{O}_j^{\mathcal{B}_A(\tau)} \widetilde{O}_j^{\mathcal{B}_B(\tau)},
\end{align}

\noindent where $\widetilde{O}_j^{\mathcal{B}_A(\tau)}:=\Gamma_j O_{\mathcal{B}_A(\tau)} \Gamma_j$, and $\widetilde{O}_j^{\mathcal{B}_B(\tau)}:=\Gamma_j O_{\mathcal{B}_B(\tau)} \Gamma_j$. Since $\Gamma_j$ are Majorana strings, they preserve the support of $O_{\mathcal{B}_A(\tau)}$ and $O_{\mathcal{B}_B(\tau)}$, which implies $\supp \left(\widetilde{O}_j^{\mathcal{B}_A(\tau)}\right)  \in \mathcal{B}_A(\tau)$ and $\supp \left(\widetilde{O}_j^{\mathcal{B}_B(\tau)}\right)  \in \mathcal{B}_B(\tau)$. Hence, it follows that one can write the action of noise channel as a convex combination of operators with the same support. Furthermore, the norm is also unchanged as $ \left\|\widetilde{O}_j^{\mathcal{B}_A(\tau)} \widetilde{O}_j^{\mathcal{B}_B(\tau)} \right \| = \left\| O_{\mathcal{B}_A(\tau)} O_{\mathcal{B}_B(\tau) } \right\|$. We can therefore write 
\begin{align}\label{eq:noise_channel_OAOB_2}
    \mathcal{U}_{\tau+1}^{\dagger} \circ \mathcal{N}_p \left(O_{\mathcal{B}_A(\tau)} O_{\mathcal{B}_B(\tau)} \right)=\sum_{j}p_j \mathcal{U}_{\tau+1}\left( \widetilde{O}_j^{\mathcal{B}_A(\tau)} \widetilde{O}_j^{\mathcal{B}_B(\tau)}\right)=\sum_j p_j O_j^{\mathcal{B_A}(\tau+1)} O_j^{\mathcal{B_B}(\tau+1)},
\end{align}

\noindent where $\supp \left(O_j^{\mathcal{B}_A(\tau+1)}\right)   \in \mathcal{B}_A(\tau+1) $ and $\supp \left(O_j^{\mathcal{B}_B(\tau+1)}\right)   \in \mathcal{B}_B(\tau+1) $, which follows from the fact that $\mathcal{U}_{\tau+1}$ is a layer consisting of strictly local unitary gates. By applying induction on the expression in Eq.~\ref{eq:noise_channel_OAOB_2} one obtains that the time-evolved operator at time $k$ can be written as Eq.~\ref{eq:expansion_OAOB(k)_local},
\begin{align}
    \left[O_AO_B\right](k)=\left[\Phi_1^{\mathrm{noisy}}\right]^{\dagger} \dots \left[\Phi_{k-1}^{\mathrm{noisy}}\right]^{\dagger}\left[\Phi_k^{\mathrm{id}}\right]^{\dagger}(O_AO_B)=\sum_{i}p_i O_i^{\mathcal{B}_A(k)} O_i^{\mathcal{B}_B(k)}.
\end{align}

\noindent Once can then bound the correlations as 
\begin{align}
    \left|\left\langle \left[O_AO_B\right](k) \right \rangle_{\rho_0} \right|&= \left|\sum_{i}p_i \left\langle O_i^{\mathcal{B}_A(k)} O_i^{\mathcal{B}_B(k)}\right \rangle_{\rho_0}\right| \leq \max_i \left |   \left\langle O_i^{\mathcal{B}_A(k)} O_i^{\mathcal{B}_B(k)}\right \rangle_{\rho_0}\right| \nonumber \\ &\leq f(d(\mathcal{B}_A(k),\mathcal{B}_B(k))  \|O_A \| \|O_B\|\leq f(d(A,B)-2vk) \|O_A \| \|O_B\|,
\end{align}
\noindent where we have used $d(\mathcal{B}_A(k),\mathcal{B_B}(k)) \geq d(A,B)-2vk$, for some constant $v=\mathcal{O}(1)$, and $\left \|O_i^{\mathcal{B}_A(k)} O_i^{\mathcal{B}_B(k)} \right \|=\|O_AO_B\|$, since the Pauli strings do not change the operator norm. This finalizes the proof of the Lemma.

\section{Asymptotic notation}\label{appendix:notation}

Throughout the paper employ the following asymptotic notation commonly used in complexity theory \cite{cormen2022introduction}:

\begin{center}
\begin{tabular}{ |c|c|c| } 
\hline
 Notation & Formal definition & Informal description \\ 
 \hline
 $f(n)=o(g(n))$ & $\forall k>0 \ \exists n_0\ \forall n>n_0: |f(n)| < kg(n)$ & $f(n)$ grows strictly slower than $g(n)$ \\  
 $f(n)=\mathcal{O}(g(n))$ & $\exists k>0\ \exists n_0\ \forall n>n_0: |f(n)| \leq kg(n)$ & $f(n)$ grows no faster than $g(n)$  \\
 $f(n)=\Omega(g(n))$ & $\exists k>0\ \exists n_0\ \forall n>n_0: |f(n)| \geq kg(n)$ & $f(n)$ grows at least as fast as $g(n)$ \\ 
 $f(n)=\Theta(g(n))$ & $\exists k_1>0, \exists k_2>0,\exists n_0\ \forall n>n_0: k_1g(n) \leq |f(n)| \leq k_2g(n)$ & $f(n)$ and $g(n)$ grow equally fast \\
 \hline
\end{tabular}
\end{center}
\end{appendix}

\end{document}